\documentclass[color-hyperref,graphics,force-numerical-bibliography,redefine-eqnarray,eucal,doublestroke,upright]{jst2018author}

\usepackage{float}

\title{Spectral properties of the singular Friedrichs-Lee Hamiltonian}

\authors{Paolo Facchi,
Marilena Ligab\`o,
Davide Lonigro}

\address[paolo.facchi@ba.infn.it]{Paolo Facchi, \\ Dipartimento di Fisica and MECENAS, Universit\`a di Bari, I-70126 Bari, Italy\\
INFN, Sezione di Bari, I-70126 Bari, Italy}

\address[marilena.ligabo@uniba.it]{Marilena Ligab\`o, \\ Dipartimento di Matematica, Universit\`a di Bari, I-70125 Bari, Italy}

\address[davide.lonigro@ba.infn.it]{Davide Lonigro, \\ Dipartimento di Fisica and MECENAS, Universit\`a di Bari, I-70126 Bari, Italy\\
INFN, Sezione di Bari, I-70126 Bari, Italy}
  \noaddress       

\theoremstyle{plain}
\newtheorem{theorem}{Theorem}[section]

\newtheorem{proposition}[theorem]{Proposition}

\newtheorem{hypothesis}{Hypothesis}

\theoremstyle{definition}
\newtheorem{definition}[theorem]{Definition}

\theoremstyle{remark}
\newtheorem{remark} [theorem]{Remark}

\newtheorem{example}[theorem]{Example}

\usepackage{color}
\usepackage{graphicx}
\usepackage{amssymb}
\usepackage{bm}
\usepackage{amsmath}
\usepackage{tikz}
\bibliography{bibliography}
\usepackage{blindtext}
\usepackage{braket}
\usepackage{footnote}

\newcommand{\R}{\mathbb{R}}
\newcommand{\CC}{\mathbb{C}}

\newcommand{\e}{\varepsilon}

\newcommand{\pvint}{\mathop{\mathrm PV}\!\int}

\begin{document}

\begin{abstract}
We show that the Friedrichs-Lee model, which describes the one-excitation sector of a two-level atom interacting with a structured boson field, can be generalized to singular atom-field couplings. We provide a characterisation of its spectrum and resonances and discuss the inverse spectral problem.
\end{abstract}


\tableofcontents 

\section{Introduction}

The Friedrichs-Lee Hamiltonian is a self-adjoint operator on a Hilbert space which describes the behaviour of an eigenvalue coupled to a continuous spectrum, and it is a rare example of a solvable model with a rich mathematical structure~\cite{Horwitz}. It was originally introduced by T.\ D.\ Lee~\cite{Lee} as a solvable quantum-field theoretical model suitable for the investigation of the renormalisation procedure. Lee's Hamiltonian has a conserved quantum number labelling reducing subspaces (excitation sectors).  Its reduction to the first nontrivial excitation sector,  which we will refer to as the Friedrichs-Lee Hamiltonian, is the quantum-mechanical model used by Friedrichs in his seminal study of the perturbation of continuous spectra~\cite{Friedrichs}.

Since its inception, the Friedrichs-Lee Hamiltonian has proven to be a very useful model in many applications, ranging from quantum field theory of unstable particles~\cite{Araki} to non-relativistic quantum electrodynamics~\cite{Frohlich}, to quantum field theory on manifolds~\cite{Turgut},  to quantum optics~\cite{Gadella}, to quantum probability~\cite{vonWaldenfels}, to name a few.

In this paper we aim at a complete study of the mathematical properties of the Friedrichs-Lee operator by extending it to a larger class of possibly singular couplings (thus providing rigorous foundations to many formal computations usually carried out in the physical literature) and providing a characterisation of its spectrum with respect to the spectrum of the uncoupled operator. The paper is organised as follows:
\begin{itemize}
	\item in Section~\ref{sect:1} we derive the expression of the Friedrichs-Lee Hamiltonian as the restriction to the one-excitation sector of Lee's field-theory model;
	\item in Section~\ref{sect:2} we introduce the singular Friedrichs-Lee Hamiltonian, proving that it includes the case of a regular coupling (Theorem~\ref{thm:singmodel}) and showing that a singularly coupled model can always be  obtained as the norm resolvent limit of a proper sequence of regular models (Theorem~\ref{thm:singlimit});
	\item in Section~\ref{sect:3} we characterise the spectrum of the Friedrichs-Lee Hamiltonian: in Theorem~\ref{thm:spessential} we find its essential and discrete components, and in Theorem~\ref{thm:spe} we find its absolutely continuous, singular continuous and pure point components, the latter being strictly dependent on a Herglotz function known as the self-energy of the model;
	\item in Section~\ref{sect:4} we apply the results of the previous sections to some examples of Friedrichs-Lee Hamiltonians;
	\item in Section~\ref{sect:5} we discuss the resonances of the model, showing that they can be characterised as complex eigenvalues of a deformation of the Hamiltonian (Theorem~\ref{thm:res});
	\item in Section~\ref{sect:6} we introduce the inverse spectral problem, i.e.\  the choice of a coupling yielding the desired dynamics for a model with given field structure.
\end{itemize}
Future developments may include the generalisation of the singular coupling and spectral characterisation to the $n$-atom Friedrichs-Lee Hamiltonian or to higher excitation sectors, as well as applications to physically interesting systems.

\section{Physical model}\label{sect:1}
Let $(X,\mu)$ be a $\sigma$-finite measure space and $\mathcal{H}=L^2(X,\mu)$ the space of square-integrable complex-valued functions with respect to $\mu$, with $\braket{\cdot|\cdot}$ denoting the inner product, conjugate-linear in the first entry, and $\|\cdot\|$  the induced norm.  Let $H_{\textrm{field}}$ be a Hamiltonian operator on the Bose-Fock space $\mathcal{F}$ of $\mathcal{H}$ with formal expression 
\begin{equation}
H_{\textrm{field}} =\int_X\omega( k )\,a^*( k )\,a( k )\;\mathrm{d}\mu( k ),
\end{equation}
where $\omega: X \to \R$ is a measurable function, and $a( k )$, $a^*( k )$ are the operator-valued distributions associated with a family of annihilation and creation operators, satisfying the formal commutation relations $[a( k ),a^*( k ')]=\delta( k - k ')$. Physically, $ H_{\textrm{field}} $ is the operator associated with the energy of a bosonic field, $(X,\mu)$  is the momentum space of the bosons, and $\omega( k )$ is the dispersion relation, i.e.\  the energy of a quantum with momentum~$k$. For example the choices $X=\mathbb{R}^3$,  $\omega( k )=(| k |^2+m^2)^{1/2}$ and $\mu$ being the Lebesgue measure on $\R^3$, describe a relativistic bosonic field associated with particles of mass $m$. At the mathematical level, $H_{\textrm{field}} $ is the second quantisation of the multiplication operator associated with the function $\omega$, and is a densely defined self-adjoint operator on $\mathcal{F}$~\cite{Reed}. 

Let us consider a nondegenerate two-level atom, with ground state (in Dirac's notation) $\ket{\downarrow}$ and excited state $\ket{\uparrow}$ in $\mathbb{C}^2$. Let
\begin{equation}
H_{\textrm{atom}}= \varepsilon_{\textrm{a}} \ket{\uparrow}\bra{\uparrow},
\end{equation}
be its Hamiltonian, where $\varepsilon_{\textrm{a}}\in\R$ is the energy of the excited state, and the ground state energy is set to zero. The operator $H_{\textrm{atom}}\otimes I+I\otimes  H_{\textrm{field}} $, defined on a dense subspace of $\mathbb{C}^2\otimes\mathcal{F}$, represents the system atom-field in the absence of mutual interaction. A physically meaningful coupling between the atom and the field can be introduced as follows: given $g\in \mathcal{H}$, let
\begin{equation}
V_g=\int_X\left(\sigma^+\otimes \overline{g( k )}\,a( k )+\sigma^-\otimes g( k )\,a^*( k )\right)\,\mathrm{d}\mu( k ),
\end{equation}
where $\sigma^+=\ket{\uparrow}\bra{\downarrow}$ and $\sigma^-=\ket{\downarrow}\bra{\uparrow}$ are the ladder operators, that is, $\sigma^+$ raises the ground to the excited state and $\sigma^-$ lowers the excited to the ground state~\cite{Frohlich}. The total Hamiltonian $H_{\mathrm{Lee}}$  associated with the atom-field system is formally given by
\begin{equation}
H_{\mathrm{Lee}}=H_{\textrm{atom}}\otimes I+I\otimes  H_{\textrm{field}}+V_g,
\end{equation}
where we use the same notation for  identity operators acting on different Hilbert spaces.
This is a generalisation of the standard Lee model~\cite{Lee}. Physically, $\mu$ controls and weighs the values of momenta available to the bosons and must be chosen according to the physical setting: for instance, for an electromagnetic field in free space, $\mu$ is the Lebesgue measure on $X=\mathbb{R}^3$, while, for a field confined in an optical cavity, at least one component of the momenta will be discrete. 

Summing up, the analytic features of our model will depend on three physically important quantities:
\begin{itemize}
	\item The space $(X,\mu)$ of all possible momenta of the field quanta;
	\item The \emph{dispersion relation} $\omega( k )$ that gives the energy of a quantum with momentum $k$;
	\item The \emph{form factor} $g( k )$ that controls the coupling between the atom and a field quantum with momentum $k$.
\end{itemize}

The operator $H_{\mathrm{Lee}}$ does not conserve the total number of bosons in the theory: the number operator, formally defined as
\begin{equation}
N_{\textrm{field}}=\int_X a^*( k ) a( k )\;\mathrm{d}\mu( k ),
\end{equation}
does not commute with $H_{\mathrm{Lee}}$ for any nonzero form factor $g$. However, the operator
\begin{equation}
N_{\mathrm{tot}}=\ket{\uparrow}\bra{\uparrow}\otimes I+ I \otimes N_{\textrm{field}},
\end{equation}
representing the total number of excitation in the system,  commutes with $H_{\mathrm{Lee}}$ for every choice of $g$; since the operator $N_{\mathrm{tot}}$ has spectrum $\sigma(N_{\mathrm{tot}})=\mathbb{N}$, one can study the evolution of the system generated by the restriction of $H_{\mathrm{Lee}}$ to each eigenspace 
of $N_{\mathrm{tot}}$.

The simplest nontrivial choice is the one-excitation sector that is isomorphic to $ \mathbb{C}\oplus \mathcal{H}$.
The generic normalised element $\Psi\in \mathbb{C}\oplus \mathcal{H}$ may be expressed as
\begin{equation}
\Psi=\begin{pmatrix}
x\\\xi
\end{pmatrix},\qquad x \in \mathbb{C}, \quad \xi \in \mathcal{H},
\end{equation}
the normalisation being
\begin{equation}
|x|^2+\int_X|\xi(k)|^2\,\mathrm{d}\mu(k)=1,
\end{equation}
where $|x|^2$ is the probability that the atom is in its excited state and $\xi$ is the wave function of the boson in the field. In particular, the state 
\begin{equation}
\Psi_0=\begin{pmatrix}
1\\0
\end{pmatrix}
\label{vacuum}
\end{equation}
represents the excited atom interacting with the vacuum. 
The restriction of $H_{\mathrm{Lee}}$ to the one-excitation sector $ \mathbb{C}\oplus \mathcal{H}$
which we will denote as $H_{\mathrm{FL}}$, is the \emph{Friedrichs-Lee Hamiltonian}~\cite{Friedrichs,Horwitz}. Its domain is
\begin{equation}
	D\left(H_{\mathrm{FL}}\right)=\mathbb{C}\oplus D(\Omega)=\left\{\begin{pmatrix}
	x\\\xi
	\end{pmatrix}\,\bigg|\,x\in\mathbb{C},\,\xi\in D(\Omega)\right\},
	\label{domain0}
\end{equation}
and it acts on a generic vector of the domain as
\begin{equation}\label{eq:Omegagrestric}
H_{\mathrm{FL}}\begin{pmatrix}
x\\\xi
\end{pmatrix}= \begin{pmatrix}\varepsilon_{\textrm{a}}  x + \langle g| \xi \rangle\\ x g + \Omega \xi\end{pmatrix},
\end{equation}
where $\Omega$ is the multiplication operator associated with the dispersion relation $\omega$, that is, $(\Omega\xi)( k )=\omega( k )\xi( k )$; we will refer to $\Omega$ as the \emph{inner Hamiltonian} of the model. 
The action of the Hamiltonian  $H_{\mathrm{FL}}$ in~\eqref{eq:Omegagrestric} can be obtained using a formal matrix representation
\begin{equation}
H_{\mathrm{FL}} = \begin{pmatrix}
\varepsilon_{\textrm{a}} &\langle g| \\
g &\Omega
\end{pmatrix},
\label{intro}
\end{equation}
where $\langle g|$ is, in Dirac notation, the linear functional on $\mathcal{H}$ associated with~$g$.  

\section{The singular Friedrichs-Lee Hamiltonian}\label{sect:2}
The Hamiltonian~\eqref{eq:Omegagrestric} with matrix representation~\eqref{intro} cannot include a \emph{singular coupling} between field and atom, i.e.\  a form factor $g\notin \mathcal{H}$. This obstruction is relevant at a physical level: for instance, a flat form factor between field and atom (i.e.\  $g(k)=\mathrm{constant}$) cannot be generally included (if $\mu$ is not a finite measure), thus preventing the description of interesting phenomena (e.g. exponential decay of the state $\Psi_0$ in Eq. \eqref{vacuum}).

To extend the Friedrichs-Lee model to a (possibly) singular coupling between atom and field~\cite{FLProc}, the formalism of Hilbert scales will be extensively used (see e.g.~\cite{AlbeverioKurasov}). 
\begin{definition}
Given a Hilbert space $\mathcal{K}$ and a self-adjoint operator $T$ on it, the space $\mathcal{K}_s$, for any $s\geq0$, is the domain of $|T|^{s/2}$ endowed with the norm 
\begin{equation}
\|\varphi\|_{s}:=\||T-i|^{s/2}\varphi\|
\label{eq:norma-s}
\end{equation}
and $\mathcal{K}_{-s}$ is its dual space, i.e.\ the space of continuous functionals on it, endowed with the norm
$\|\varphi\|_{-s}:=\||T-i|^{-s/2}\varphi\|$.
\end{definition}
With an abuse of notation, the  pairing between $\varphi\in\mathcal{K}_{-s}$ and $\zeta \in\mathcal{K}_s$ will still be  denoted as $\braket{\varphi|\zeta}$, i.e.\ with the notation of the scalar product on $\mathcal{K}$. Besides, the following properties hold:
\begin{itemize}
	\item $\mathcal{K}_{s}\subset\mathcal{K}_{s'}$ for any $s>s'$, with dense inclusion;
	\item $\mathcal{K}_0$, $\mathcal{K}_1$ and $\mathcal{K}_2$ are, respectively, the original Hilbert space $\mathcal{K}$, the form domain of $T$, and the domain of $T$;
	\item $T$ maps $\mathcal{K}_s$ into $\mathcal{K}_{s-2}$ and $\frac{1}{T-i}$ maps $\mathcal{K}_s$ into $\mathcal{K}_{s+2}$.
\end{itemize}

In the following we will consider the Hilbert space $\mathcal{K}=\mathcal{H}$ 
and the self-adjoint operator $T=\Omega$ as the multiplication operator by a real-valued measurable function $\omega: X \to \R$.\footnote{There is no loss of generality in this choice since, by the spectral theorem, every self-adjoint operator is equivalent to a multiplication operator on some $L^2$ space.} In this case we have, for every $s\geq0$,
\begin{equation}
	g\in\mathcal{H}_s\qquad \text{iff}\qquad \int_X\left|\omega(k)-i\right|^s|g(k)|^2\,\mathrm{d}\mu(k)<\infty;
\end{equation}
\begin{equation}
	g\in\mathcal{H}_{-s}\qquad \text{iff}\qquad \int_X\frac{|g(k)|^2}{|\omega(k)-i|^s}\,\mathrm{d}\mu(k)<\infty.
\end{equation}
Our first result is the following theorem, which extends the Friedrichs-Lee Hamiltonian to the case $g \in \mathcal{H}_{-2}$.
\begin{theorem}\label{thm:singmodel}
	Let $\varepsilon\in\mathbb{R}$, $\Omega$ be a self-adjoint operator on the Hilbert space $\mathcal{H}=L^2(X,\mu)$ acting as the multiplication operator by a real-valued measurable function $\omega:X \to \R$, and $g\in\mathcal{H}_{-2}$. Consider the operator $H_{g,\varepsilon}$ on the Hilbert space $\CC \oplus \mathcal{H}$ with domain
	\begin{equation}
	\label{domain}
	D(H_{g,\varepsilon})=\left\{\begin{pmatrix}
	x\\\xi-x\frac{\Omega}{\Omega^2+1}g
	\end{pmatrix}\Bigg|\;x\in\mathbb{C},\;\xi\in\mathcal{H}_2\right\} ,
	\end{equation}
	such that
	\begin{equation}
	\label{action}
	H_{g,\varepsilon}\begin{pmatrix}
	x\\\xi-x\frac{\Omega}{\Omega^2+1}g
	\end{pmatrix}=\begin{pmatrix}
	\varepsilon x+\braket{g|\xi}\\
	\Omega\xi+x\frac{1}{\Omega^2+1}g
	\end{pmatrix}.
	\end{equation}
	Then we have:
	\begin{enumerate}
		\item[(i)] If $g\in\mathcal{H}$, then $H_{g,\varepsilon}$ reduces to the Friedrichs-Lee Hamiltonian $H_{\mathrm{FL}}$ in~(\ref{eq:Omegagrestric}) with atom excitation energy
		\begin{equation}
		\varepsilon_{\mathrm{a}}=\varepsilon+\left\langle g\,\bigg|\frac{\Omega}{\Omega^2+1}g\right\rangle.
		\label{shift}
		\end{equation}
		\item[(ii)] $H_{g,\varepsilon}$ is self-adjoint and, for all $z \in \CC \setminus \R$, its resolvent operator is
		\begin{equation}
		\label{resolvent00}
		\frac{1}{H_{g,\varepsilon}-z}\begin{pmatrix}
		x\\\xi
		\end{pmatrix}=\begin{pmatrix}
		\frac{x-\left\langle g\,\big|\frac{1}{\Omega-z}\xi\right\rangle}{\varepsilon-z-\Sigma_g(z)}\\\frac{1}{\Omega-z}\xi-\frac{x-\left\langle g\,\big|\frac{1}{\Omega-z}\xi\right\rangle}{\varepsilon-z-\Sigma_g(z)}\frac{1}{\Omega-z}g
		\end{pmatrix}, \qquad  \begin{pmatrix}
		x\\\xi
		\end{pmatrix} \in \CC \oplus \mathcal{H},
		\end{equation}
		where 
		\begin{equation}
		\Sigma_g(z):=\left\langle g\,\bigg|\left(\frac{1}{\Omega-z}-\frac{\Omega}{\Omega^2+1}\right)g\right\rangle\label{self}
		\end{equation}
		is the self-energy function of $H_{g,\varepsilon}$.
		\item[(iii)] The evolution group generated by $H_{g,\varepsilon}$ is given by
		\begin{equation}\label{eq:unitarygroup}
		U_{H_{g,\varepsilon}}(t)=\frac{1}{2\pi i}\pvint_{i \delta-\infty}^{i \delta+\infty}e^{-izt}\,\frac{1}{H_{g,\varepsilon}-z}\;\mathrm{d}z,
		\end{equation}
		for all  $t>0$, where $\delta >0$ is an arbitrary constant and the principal-value integral must be understood in the strong sense. Moreover, if 
		\begin{equation}\label{psiev}
		\begin{pmatrix}
		x(t)\\ \xi(t)
		\end{pmatrix}
		=U_{H_{g,\varepsilon}}(t) \Psi_0
		\end{equation}
		is the evolution at time $t>0$ of the initial state $\Psi_0$ in~(\ref{vacuum}), we have
		\begin{equation}\label{f0t}
		x(t)=\frac{1}{2\pi i}\pvint_{i \delta-\infty}^{i\delta+\infty}\frac{e^{-izt}}{ \varepsilon -z-\Sigma_g(z)}\;\mathrm{d}z,
		\end{equation}
		\begin{equation}\label{xit}
		\xi(t,k)=-\frac{1}{2\pi i}\pvint_{i \delta-\infty}^{i\delta+\infty}\frac{g(k)\, e^{-izt}}{ [\varepsilon -z-\Sigma_g(z)] [\omega(k)-z]} \;\mathrm{d}z.
		\end{equation}
	\end{enumerate}
\end{theorem}

	\begin{proof}
	$(i)$ If $g\in\mathcal{H}$, then $\frac{\Omega}{\Omega^2+1}g\in \mathcal{H}_{2}=D(\Omega)$ and hence the domains in Eqs.~\eqref{domain0} and~\eqref{domain} coincide. Applying $H_{\mathrm{FL}}$ to any vector of the form~\eqref{domain} yields the same result as in Eq.~\eqref{action}, hence the two operators coincide.
	
	$(ii)$ Since $D(\Omega)$ is dense in $\mathcal{H}$, then $D(H_{g,\varepsilon})$ is dense in $\mathbb{C}\oplus\mathcal{H}$. Moreover, the self-energy function $\Sigma_g$ in Eq.~\eqref{self} is well defined and a direct calculation shows that $H_{g,\varepsilon}$ is symmetric. Finally the bounded operator acting on $\mathbb{C}\oplus\mathcal{H}$ as in Eq.~\eqref{resolvent00} is the inverse of $H_{g,\varepsilon}-z$ for any $z\in\CC\setminus\R$ and the  self-adjointess of $H_{g,\varepsilon}$ easily follows  from that.
	
	$(iii)$ Eq.~\eqref{eq:unitarygroup} follows from Eq.~\eqref{resolvent00} and from the  link between the resolvent and the evolution group associated with any self-adjoint operator~\cite{Exner}; Eq.~\eqref{f0t}, in particular, follows by substituting $x=1$ and $\xi=0$ in Eq.~\eqref{resolvent00} and applying Eq.~\eqref{eq:unitarygroup}.
	\end{proof}

\begin{remark}
\label{rmk:3.3}
We can distinguish three separate cases:\begin{enumerate}
	\item[1.] $g\in\mathcal{H}$: the domain  $D(H_{g,\e})$ does not depend on $g$, and both $\varepsilon_{\textrm{a}}$ and $\varepsilon$ are finite quantities, representing respectively the ``bare" and ``dressed" (coupling-dependent) excitation energy of the atom. The formal matrix expression~\eqref{intro} of the Hamiltonian holds.
	\item[2.] $g\in\mathcal{H}_{-1}\setminus\mathcal{H}$: the domain $D(H_{g,\e})$ depends on $g$, but again both $\varepsilon_{\textrm{a}}$ and $\varepsilon$ are finite quantities with the same physical meaning as above, since $\braket{g\big|\frac{\Omega}{\Omega^2+1}g}$ is finite. Again the Hamiltonian can be written as in Eq.~\eqref{intro}.
	\item[3.] $g\in\mathcal{H}_{-2}\setminus\mathcal{H}_{-1}$: the domain $D(H_{g,\e})$ depends on $g$, and the bare excitation energy $\varepsilon_{\textrm{a}}$ is not defined, since $\braket{g\big|\frac{\Omega}{\Omega^2+1}g}$ is not finite; because of that, Eq.~\eqref{intro} is  ill-defined.
\end{enumerate}
The latter situation is related to the renormalisation procedure of quantum field theory, in which the bare (and hence unobservable) value of a parameter, e.g. the electron charge, diverges in such a way to obtain a finite value of the measurable dressed one. Besides, in the first two cases we may equivalently write
\begin{equation}\label{eqn:barese}
	\varepsilon-\Sigma_g(z)=\varepsilon_{\textrm{a}}-\tilde{\Sigma}_{g}(z)
\end{equation}
where $\tilde{\Sigma}_{g}(z)=\braket{g\big|\frac{1}{\Omega-z}g}$ is the ``bare'' self-energy. In this sense, the extension of the model to the case $g\in\mathcal{H}_{-1}$ is straightforward up to an algebraic technicality, i.e.\  the choice of a convenient representation of the domain, while the further extension to the case $g\in\mathcal{H}_{-2}$ requires an ``infinite" term to be added to both the bare excitation energy and the bare self-energy.

Finally, the three cases reflect the possible situations in which $\Psi_0$ has
\begin{enumerate}
	\item[1.] finite mean value and variance of energy $H_{g,\varepsilon}$;
	\item[2.] finite mean value, but infinite variance of  $H_{g,\varepsilon}$;
	\item[3.] infinite mean value and variance of  $H_{g,\varepsilon}$.
\end{enumerate}
\end{remark}
When $g\notin\mathcal{H}$, we will say that the atom-field coupling is \emph{singular}, as opposed to the regular case $g\in\mathcal{H}$. We would like to obtain a singular Friedrichs-Lee Hamiltonian as the limit, in a suitable sense, of a sequence of regular models. The usual notions of norm and strong operator convergence cannot apply because of the unboundedness of the operators and the fact that the domain of the singular Hamiltonian is coupling-dependent; however, we can resort to the notions of \emph{resolvent} and \emph{dynamical} convergence. 

Recall that, given a family $(T_n)_{n\in\mathbb{N}}$ of self-adjoint operators on a Hilbert space $\mathcal{K}$ and a self-adjoint operator $T$, the sequence $(T_n)_{n\in\mathbb{N}}$ is said to converge to $T$:
\begin{itemize}
	\item in the \emph{norm} (resp. \emph{strong}) \emph{resolvent sense} if, for all $z\in\CC\setminus\R$ $\frac{1}{T_n-z}\to\frac{1}{T-z}$ in the norm (resp. strong) sense, as $n \to \infty$;
	\item in the \emph{norm} (resp. \emph{strong}) \emph{dynamical sense} if, for all $t\in\mathbb{R}$, $e^{-itT_n}\to e^{-itT}$ in the norm (resp. strong) sense, as $n \to \infty$.
\end{itemize}
The following result holds:
\begin{theorem}[Singular coupling limit]\label{thm:singlimit}
Let $\e \in \R$.
	\begin{itemize}
		\item [(i)] If $(g_n)_{n\in\mathbb{N}}\subset\mathcal{H}$ is a convergent sequence in the norm of $\mathcal{H}_{-2}$ with limit $g \in \mathcal{H}_{-2}$, then $H_{g_n,\varepsilon}\to H_{g,\varepsilon}$ in the norm resolvent sense and in the strong dynamical sense as $n\to\infty$.
		\item [(ii)] Conversely, for every singular Friedrichs-Lee Hamiltonian $H_{g,\varepsilon}$, $g \in \mathcal{H}_{-2}$, there exists a sequence $(g_n)_{n\in\mathbb{N}}\subset \mathcal{H}$ such that $H_{g_n,\varepsilon}\to H_{g,\varepsilon}$ in the norm resolvent sense and in the strong dynamical sense as $n\to\infty$.
	\end{itemize}
\end{theorem}
\begin{proof}
	First of all, notice that norm resolvent convergence obviously implies strong resolvent convergence, and the latter is \emph{equivalent} to strong dynamical convergence (see e.g.~\cite{DeOliveira}), hence we only need to prove the results about norm resolvent convergence.
	
	$(i)$ A direct calculation shows that the resolvent operator in Eq.~\eqref{resolvent00},  with form factor $g_n\in\mathcal{H}$ and $g\in\mathcal{H}_{-2}$ respectively, can be written as the sum of the resolvent of the uncoupled operator $H_{0,\e}$ plus a finite-rank term, namely,
	\begin{equation}
\frac{1}{H_{g_n,\e}-z}=\frac{1}{H_{0,\e}-z}+ K_n,	
        \end{equation}
	\begin{equation}
	\frac{1}{H_{g,\e}-z}=\frac{1}{H_{0,\e}-z}+ K,	
	\end{equation}
	where for all $\Psi=\begin{pmatrix} x \\ \xi \end{pmatrix} \in \mathbb{C} \oplus \mathcal{H}$
	\begin{equation}
	K_n \Psi= \frac{1}{\e-z-\Sigma_{g_n}(z)} \begin{pmatrix}
		x \frac{ \Sigma_{g_n}(z)}{\e-z} - \braket{\frac{1}{\Omega-\bar{z}}g_n | \xi } \\
		 x \frac{1 }{\Omega-z } g_n + \braket{\frac{1}{\Omega - \bar{z}} g_n | \xi} \frac{1}{\Omega - z} g_n
		\end{pmatrix}
	\end{equation}
	and
	\begin{equation}
	K \Psi= \frac{1}{\e-z-\Sigma_{g}(z)} \begin{pmatrix}
		x \frac{ \Sigma_{g}(z)}{\e-z} - \braket{\frac{1}{\Omega-\bar{z}}g | \xi } \\
		 x \frac{1 }{\Omega-z } g + \braket{\frac{1}{\Omega - \bar{z}} g| \xi} \frac{1}{\Omega - z} g
		\end{pmatrix}.
	\end{equation}
	Hence the claim is proven if we show that for all $z\in\mathbb{C}\setminus\R$
	\begin{itemize}
		\item $\left \|\frac{1}{\Omega-z}g_n - \frac{1}{\Omega-z}g \right \| \to 0$, as $n \to \infty$;
		\item $\Sigma_{g_n}(z)\to\Sigma_g(z)$, as $n \to \infty$.
	\end{itemize}	
	We recall that $g_n\to g$,  as $n \to \infty$, in the norm of $\mathcal{H}_{-2}$, means that $\|\frac{1}{\Omega-i}(g_n-g)\|\to0$, and equivalently, by the first resolvent formula, $\|\frac{1}{\Omega-z}g_n - \frac{1}{\Omega-z}g\| \to 0$, for every $z\in\mathbb{C}\setminus\mathbb{R}$; moreover, by continuity of the pairing between $\mathcal{H}_2$ and $\mathcal{H}_{-2}$, this also proves $\Sigma_{g_n}(z)\to\Sigma_g(z)$.
	
	$(ii)$ Let $g\in\mathcal{H}_{-2}$. Since $\mathcal{H}$ is dense in $\mathcal{H}_{-2}$, there exists a sequence $(g_n)_{n\in\mathbb{N}}$ such that $g_n\to g$, as $n\to\infty$, in the norm of $\mathcal{H}_{-2}$, and hence, as in $(i)$, a sequence of regular Friedrichs-Lee Hamiltonians which converges in the norm resolvent sense to the singular one.
\end{proof}
\begin{remark}
If $g\in\mathcal{H}_{-2}$, the approximating sequence of Hamiltonians $H_{g_n, \e}$ is characterised by a diverging bare excitation energy $\varepsilon_{\mathrm{a},n}= \e+ \braket{g_n | \frac{\Omega}{\Omega^2+1} g_n}$ \emph{and} a diverging bare self-energy $\tilde{\Sigma}_{g_n}(z) = \braket{g_n | \frac{1}{\Omega-z} g_n}$; their difference converges to a finite limit which depends on the value of the dressed excitation energy $\varepsilon$; this is clearly a renormalisation procedure, as discussed  in Remark~\ref{rmk:3.3}.

Also notice that  Theorem~\ref{thm:singlimit} holds even if the dressed energy $\varepsilon$ of the approximating sequence, instead of being kept fixed, is replaced with a converging sequence $\varepsilon_n\to\varepsilon$.
\end{remark}
\begin{remark}
	There is an interesting connection between the Friedrichs-Lee model and rank-one perturbations of self-adjoint operators. Given a Hilbert space $\mathcal{K}$, a self-adjoint operator $T$ on $\mathcal{K}$, and a functional $\varphi\in\mathcal{K}_{-2}$, consider the  formal operator
	\begin{equation}
	T+\alpha\ket{\varphi}\bra{\varphi}, \qquad \alpha \in \R.
	\label{formal}
	\end{equation}
	If $\varphi\in\mathcal{K}_{-2}\setminus\mathcal{K}$, this is only a formal expression, with which one can associate a well-defined self-adjoint operator through a restriction-extension procedure~\cite{AlbeverioKurasov,SimonRankOne2,Posilicano2003,Posilicano2008} that we  briefly recall in the following. Consider the densely defined symmetric operator $T_\varphi$ obtained by restricting $T$ to $\ker(\bra{\varphi})$; one can prove that its adjoint $T^*_\varphi$ is the operator with domain
	\begin{equation}
	D(T^*_{\varphi})=\left\{\xi-x\frac{T}{T^2+1}\varphi\,\bigg|\,\xi\in D(T),\,x\in\CC\right\}
	\end{equation}
	acting as
	\begin{equation}
	T^*_{\varphi}\left(\xi-x\frac{T}{T^2+1}\varphi\right)=T\xi+x\frac{1}{T^2+1}\varphi.
	\end{equation}
	Then, for $\varphi\in\mathcal{K}_{-1}\setminus\mathcal{K}$, the restriction $T_{\varphi,\alpha}$ of $T^*_\varphi$ to the domain
	\begin{equation}
		D(T_{\varphi,\alpha})=\left\{\xi-x\frac{T}{T^2+1}\varphi \,\bigg|\, \xi\in D(T),\,x\in\CC, \,
		\braket{\varphi|\xi}=-x\left(\frac{1}{\alpha}+c_\varphi\right)\right\},
	\end{equation}
	with
	$c_\varphi=\langle\varphi |\frac{T}{T^2+1}\varphi\rangle$,
	is a self-adjoint operator which corresponds to a singular rank-one perturbation of $T$: indeed, if one applies the formal expression~\eqref{formal} on vectors in the above domain, all terms outside the Hilbert space $\mathcal{K}$ cancel out because of the constraint in the domain definition.	
	
	If, instead, $\varphi\in\mathcal{K}_{-2}\setminus\mathcal{K}_{-1}$, there is an issue: the action of $\bra{\varphi}$ on $\frac{T}{T^2+1}\varphi$ is not defined. However, we can still fix some arbitrary $c\in\mathbb{R}$ and define $T^c_{\varphi,\alpha}$ as the restriction of $T^*_\varphi$ to the domain
	\begin{equation}
	D(T^c_{\varphi,\alpha})=\left\{\xi-x\frac{T}{T^2+1}\varphi \,\bigg|\, \xi\in D(T),\,x\in\CC, \,
	\braket{\varphi|\xi}=-x\left(\frac{1}{\alpha}+c\right)\right\}.
	\end{equation}
	
	Interestingly, in both cases, the dependence on the parameter $\alpha$ is in the domain of the operator, and not in its action. Besides, in the case $\varphi\in\mathcal{K}_{-2}\setminus\mathcal{K}_{-1}$, the operator depends on $1/\alpha+c$ rather than on $\alpha$ and $c$ separately.
	
	Now we present an alternative procedure obtained by introducing an extra degree of freedom in the system, i.e.\ by searching for self-adjoint realisations on the larger Hilbert space $\mathbb{C}\oplus\mathcal{K}$, rather than on $\mathcal{K}$. Let us define the operator $\tilde{T}^c_{\varphi,\alpha}$ with domain
	\begin{equation}
		D(\tilde{T}^c_{\varphi,\alpha})=\left\{\begin{pmatrix}
x\\\xi-x\frac{T}{T^2+1}\varphi
		\end{pmatrix}\,\Bigg|\,\xi\in D(T),\,x\in\CC\right\}
	\end{equation}
	acting as
	\begin{equation}
\tilde{T}^c_{\varphi,\alpha}\begin{pmatrix}
x\\\xi-x\frac{T}{T^2+1}\varphi
\end{pmatrix}=\begin{pmatrix}
\left(\frac{1}{\alpha}+c\right)x+\braket{\varphi|\xi}\\
T\xi+x\frac{1}{T^2+1}\varphi
\end{pmatrix}.
	\end{equation}
	In practice, we are handling the ``diverging" terms by enlarging the Hilbert space instead than  imposing a constraint: the dependence on $\alpha$ (and on $c$) is therefore moved to the action of the operator rather than to its domain.
	
	Now, by choosing $\mathcal{K}=\mathcal{H}$, $\Omega$ as the multiplication operator by $\omega$, and $\varphi=g$, the operator $\tilde{T}^c_{\varphi,\alpha}$ corresponds to a singular Friedrichs-Lee Hamiltonian $H_{g,\e}$ with dressed excitation energy
	\begin{equation}
		\varepsilon=\frac{1}{\alpha}+c.
	\end{equation}
	Notice that the freedom in the choice of $c$ reflects the fact that, in the Friedrichs-Lee model, a bare excitation energy is not defined for $g\notin\mathcal{H}_{-1}$: the operator really depends only on $1/\alpha+c$, in the same way as the Friedrichs-Lee Hamiltonian ultimately depends on the dressed energy $\varepsilon$ alone.
\end{remark}

\section{Spectral properties}  \label{sect:3}
After having introduced the model, let us characterise its spectral properties with respect to two common decompositions of the spectrum of a self-adjoint operator:
\begin{itemize}
	\item absolutely continuous, singular continuous and pure point spectrum;
	\item essential and discrete spectrum,
\end{itemize}
(the discrete spectrum being the set of all isolated eigenvalues of finite multiplicity~\cite{DeOliveira}).

In the absence of coupling (i.e.\  for $g=0$), the spectrum of the Friedrichs-Lee Hamiltonian $H_{0,\e}$ is obviously
\begin{equation}
	\sigma(H_{0,\e})=\{\varepsilon\}\cup\sigma(\Omega),
\end{equation}
with $\sigma(\Omega)$ being the spectrum of $\Omega$, i.e.\  the closure of the $\mu$-essential range of $\omega$, 
\begin{equation}
\mu\operatorname{-\, Ran}(\omega)= \left\{\lambda \in \R \,|\, \mu\bigl( \omega^{-1} ((\lambda - \delta, \lambda + \delta))\bigr)>0, \quad \forall \delta>0\right\},
\end{equation}
which coincides with the support of the induced measure $\nu_{\Omega}$ defined as 
\begin{equation}
	\nu_{\Omega}(B)=\int_{\omega^{-1}(B)}\mathrm{d}\mu(k),
\end{equation}
for every Borel set $B\subset\mathbb{R}$. 

Here and henceforth the \emph{support of a Borel measure} $\nu$ on $\R$ is the set
\begin{equation}
\label{eq:suppdef}
\operatorname{supp}(\nu) = \left\{\lambda \in \R \,|\, \nu\bigl((\lambda - \delta, \lambda + \delta)\bigr)>0, \quad \forall \delta>0\right\}.
\end{equation}
This is the smallest closed set $C$ such that $\nu(\R\setminus C) =0$, and is also known as  the set of growth points of $\nu$ or the spectrum of $\nu$. 

Moreover, by abusing English language with the use of an adjective to modify rather to limit a noun, we define a \emph{minimal support} (or an essential support) of the measure~$\nu$ a set $M$ such that $\nu(\R\setminus M) =0$, and such that every Borel subset $M_0\subset M$ with $\nu(M_0)=0$ has also zero Lebesgue measure. Notice that a minimal support is not unique: minimal supports may differ by sets of zero Lebesgue and $\nu$ measure, and it can happen either that $\operatorname{supp}(\nu)$ is a minimal support or that it is not: there  exists always a minimal support of $\nu$ whose closure coincides with $\operatorname{supp}(\nu)$, but, in general, the closure of a minimal support may differ from $\operatorname{supp}(\nu)$ by a set of nonzero Lebesgue measure.

The absolutely continuous (ac), singular continuous (sc) and pure point (pp) components of $\sigma(\Omega)$ coincide with the supports of the ac, sc and pp components of $\nu_\Omega$, whence
\begin{itemize}
	\item $\sigma_{\mathrm{ac}}(H_{0,\e})=\sigma_{\mathrm{ac}}(\Omega)$;
	\item $\sigma_{\mathrm{sc}}(H_{0,\e})=\sigma_{\mathrm{sc}}(\Omega)$;
		\item $\sigma_{\mathrm{pp}}(H_{0,\e})=\{\varepsilon\}\cup\sigma_{\mathrm{pp}}(\Omega)$,
\end{itemize}
with $\Psi_0$ being the eigenvector in Eq.~\eqref{vacuum} belonging to $\varepsilon$. At the physical level, $\operatorname{supp}(\nu_\Omega)=\sigma(\Omega)$ is the energy space of the boson. As for the distinction between essential and discrete spectrum, in the most general case we have
\begin{itemize}
	\item $\sigma_{\mathrm{ess}}(H_{0,\e})=\sigma_{\mathrm{ess}}(\Omega)$;
	\item $\sigma_{\mathrm{dis}}(H_{0,\e})\setminus\{\e\}=\sigma_{\mathrm{dis}}(\Omega)$,
\end{itemize}
with $\e$ belonging to the discrete spectrum  $\sigma_{\mathrm{dis}}(H_{0,\e})$ if and only if
$\e$ is isolated from the spectrum of $\Omega$.

We want to find a complete characterisation of the spectral properties of $H_{g,\e}$ with respect to the spectrum of $\Omega$ for a generic form factor $g \in \mathcal{H}_{-2}$. First of all, let us examine the behavior of the discrete/essential decomposition of the spectrum.
\begin{theorem}
Let $\e \in \R$ and $g \in \mathcal{H}_{-2}$. The essential spectrum of $H_{g,\e}$ in~\eqref{domain}--\eqref{action}
coincides with the essential spectrum of $\Omega$, with the  possible exception of 
the accumulation points of the eigenvalues of $\Omega$.
\label{thm:spessential}
\end{theorem}
\begin{proof}
	Suppose $g\in\mathcal{H}$, i.e.\  consider the regular model. Then, using the matrix representation for $H_{g,\e}$, we can write
	\begin{equation}
		H_{g,\e}=\begin{pmatrix}
		\e&0\\
		0&\Omega
		\end{pmatrix}+\begin{pmatrix}
		\varepsilon_{\textrm{a}}-\e&\bra{g}\\
		g&0
		\end{pmatrix}\equiv H_{0,\e}+K_{g,\e},
	\end{equation}
	with $\varepsilon_{\textrm{a}}$ as in Eq.~\eqref{shift}. $K_{g,\e}$ is finite-rank and hence leaves the essential spectrum unchanged, thus $\sigma_{\mathrm{ess}}(H_{g,\e})=\sigma_{\mathrm{ess}}(H_{0,\e})=\sigma_{\mathrm{ess}}(\Omega)$. 
	
	If $g\in\mathcal{H}_{-2}\setminus\mathcal{H}$, Theorem~\ref{thm:singlimit} ensures that $H_{g,\e}$ will be the norm resolvent limit of a sequence of regular models sharing the same essential spectrum. Under these conditions, the norm resolvent limit preserves the essential spectrum with the possible exception of accumulation points of the eigenvalues of $\Omega$~\cite{DeOliveira}. This proves the claim.
\end{proof}
\begin{remark}
	As a consequence of Theorem~\ref{thm:spessential}, when $\sigma_{\mathrm{ess}}(\Omega)$ is entirely continuous or dense pure point, it will coincide with the essential spectrum of the corresponding Friedrichs-Lee operator. However, there may be conversion of dense pure point spectrum into continuous spectrum or vice versa.
\end{remark}
Now let us study the decomposition of the spectrum $\sigma(H_{g,\e})$ into its absolutely continuous, singular continuous and pure point components. We define the \emph{coupling measure} associated with $\Omega$ and $g$ as 
\begin{equation}
	\kappa_g(B)=\int_{\omega^{-1}(B)}|g(k)|^2\,\mathrm{d}\mu(k),
		\label{eqn:nug}
\end{equation}
for all Borel sets $B \subset \mathbb{R}$.

 We have the following result:
\begin{theorem}
Let $\e \in \R$ and $g \in \mathcal{H}_{-2}$. Let $\mathcal{H}_g=L^2(X_g,\mu)$, with
\begin{equation}\label{eqn:xg}
X_g=\omega^{-1}\left(\operatorname{supp}(\kappa_g)\right),
\end{equation}
and let $\mathcal{H}_g^\perp$ be its orthogonal complement. Let $G_g:\R\rightarrow~\mathbb{R}~\cup~\{\infty\}$ be the function \begin{equation}
	G_g(\lambda)=\int_\R\frac{1}{(\lambda-\lambda')^2}\,\mathrm{d}\kappa_g(\lambda'),
	\label{eqn:gi}
\end{equation}
and let $\Sigma_g^+$ be   the boundary value on the real axis of the self-energy function $\Sigma_g$ in~\eqref{self}, defined almost everywhere by
\begin{equation}
\Sigma^+_g(\lambda)=\lim_{\delta\downarrow0}\Sigma_g(\lambda+i\delta).
\label{eq:Sigma+def}
\end{equation}
Then
\begin{itemize}
	\item $\sigma_{\mathrm{ac}}(H_{g,\varepsilon})=\sigma_{\mathrm{ac}}(\Omega)$;
		\item $\sigma_{\mathrm{pp}}(H_{g,\varepsilon})=\sigma_{\mathrm{pp}}\left(\Omega\big|_{\mathcal{H}_g^\perp}\right)\cup\overline{\left\{\lambda\in \R\,\big|\,\varepsilon-\lambda=\Sigma^+_g(\lambda),\,G_g(\lambda)<\infty\right\}}$;
		\item $\sigma_{\mathrm{sc}}\left(H_{g,\varepsilon}\big|_{\mathcal{H}_g^\perp}\right)= \sigma_{\mathrm{sc}}\left(\Omega\big|_{\mathcal{H}_g^\perp}\right)$;
\item the set $\left\{\lambda\in \R\,\big|\,\varepsilon-\lambda=\Sigma^+_g(\lambda),\,G_g(\lambda)=\infty\right\}$ is a minimal support for a maximal singular continuous measure of $H_{g,\varepsilon}$;
\item the restrictions of $\Omega$ and $H_{g,\e}$ to their absolutely continuous subspaces are unitarily equivalent.
\end{itemize}
\label{thm:spe}
\end{theorem}
\begin{remark}
This result can be explained as follows. First of all, the space $X$ of field momenta can be split into a subset $X_g$ of momenta which are effectively coupled to the atom (i.e.\  on which the form factor $g$ is $\mu$-supported) and a complementary subset of uncoupled momenta; this subdivision induces a correspondent subdivision of the energy space (i.e.\  the support of $\nu_\Omega$) into coupled and uncoupled energies, i.e.\  the support of $\kappa_g$ and its complement. 

As expected, the uncoupled part of the spectrum is independent of $g$ and hence, in particular, is the same as in the case $g=0$. As for the coupled one, it turns out that the absolutely continuous spectrum is still unchanged, but the singular (i.e.\  pure point and singular continuous) spectrum will be minimally supported by the set of solutions of the equation
\begin{equation}
	\varepsilon-\lambda=\Sigma^+_g(\lambda).
	\label{eigeneq}
\end{equation}
Finally, notice that the pole equation~\eqref{eigeneq} admits the unique solution $\lambda=\e$ only when $g=0$ and, in this case, necessarily $G(\lambda)=0$, hence $\e$ is in the pure point spectrum; our result is in full agreement with the uncoupled case discussed above.
\end{remark}

\begin{remark}
If one substitutes $g$ with $\beta g$ for some $\beta\in\R$, and hence $\Sigma_{\beta g}(z)=\beta^2\Sigma_g(z)$, the singular spectrum becomes supported by the set of solutions of the equation $\Sigma_g^+(\lambda)=\frac{\varepsilon-\lambda}{\beta^2}$, which will be a different set for every value of $\beta$. Therefore the singular spectrum of the Friedrichs-Lee Hamiltonian $H_{g,\e}$  is highly coupling-dependent. 
\end{remark}

The proof of Theorem~\ref{thm:spe}  will  be given in section~\ref{sec:spe}. First we need some mathematical preliminaries.

\subsection{The self-energy as a Borel transform}
\begin{definition}
	Let $\nu$ be a Borel measure on $\R$ satisfying the growth condition
	\begin{equation}
		\int_{\mathbb{R}}\frac{1}{1+\lambda^2}\,\mathrm{d}\nu(\lambda)<\infty.
		\label{growth}
	\end{equation}
	Its (regularised) \emph{Borel transform} is the  function $B_\nu:\CC\setminus\operatorname{supp}(\nu)\rightarrow\CC$ acting as
	\begin{equation}
		B_\nu(z)=\int_\R\left(\frac{1}{\lambda-z}-\frac{\lambda}{1+\lambda^2}\right)\mathrm{d}\nu(\lambda).
		\label{eqn:regbor}
	\end{equation}
\end{definition}
\begin{remark}
	In the literature, the (standard) Borel transform of a Borel measure $\nu$ on $\R$ is often defined as follows:
	$$
	\tilde{B}_\nu(z)=\int_\R\frac{1}{\lambda-z}\,\mathrm{d}\nu(\lambda).
	$$
	However, this definition makes sense for a smaller class of measure, i.e.\  measures $\nu$ satisfying the growth condition
	\begin{equation}
	\int_\R\frac{1}{1+|\lambda|}\,\mathrm{d}\nu(\lambda)<\infty.
	\label{growth0}
	\end{equation}For such measures, $\tilde{B}_\nu(z)$ and $B_\nu(z)$ only differ by a finite real constant; this difference is in fact immaterial for our purposes, since, as we will show in the next proposition, $\nu$  depends only on the imaginary part of the boundary values of $B_\nu(z)$ on the real line, but the choice~\eqref{eqn:regbor} is more convenient since it is well-defined for a larger class of measures. Indeed, for a Friedrichs-Lee Hamiltonian with form factor $g$, we have
		\begin{itemize}
			\item $g\in\mathcal{H}\qquad\text{iff}\qquad \int_\mathbb{R}\mathrm{d}\kappa_g(\lambda)< \infty$ ;
			\item $g\in\mathcal{H}_{-1}\qquad\text{iff}\qquad\int_\mathbb{R}\frac{1}{1+|\lambda|}\,\mathrm{d}\kappa_g(\lambda) < \infty$;
			\item $g\in\mathcal{H}_{-2}\qquad\text{iff}\qquad\int_\mathbb{R}\frac{1}{1+\lambda^2}\,\mathrm{d}\kappa_g(\lambda)< \infty$,
		\end{itemize}
	and hence, in particular, the case $g\in\mathcal{H}_{-2}\setminus\mathcal{H}_{-1}$ corresponds to the case in which $\kappa_g$ does not admit a standard Borel transform (i.e.\  the bare self-energy $\tilde{\Sigma}_{g}(z)$), but does have a regularised Borel transform.
\end{remark}
\begin{proposition}
	Let $g \in \mathcal{H}_{-2}$. The self-energy $\Sigma_g$ of a Friedrichs-Lee Hamiltonian defined in Eq.~\eqref{self} has the following properties:
	\begin{itemize}
		\item [(i)] $\Sigma_g$ is analytic in $\mathbb{C}\setminus\operatorname{supp}(\kappa_g)$ and $\mathrm{Im}\,\Sigma_g(z)>0$ for all $z\in\mathbb{C}^+$, i.e.\ it is a Herglotz function~\cite{Gesztesy2};
		
		\item [(ii)] the coupling measure $\kappa_g$ in~\eqref{eqn:nug} can be uniquely reconstructed from the imaginary part of $\Sigma_g$ by Stieltjes' inversion formula:
		\begin{equation}
			\frac{1}{2}\Bigl(\kappa_g\bigl((\lambda_0,\lambda)\bigr) + \kappa_g\bigl([\lambda_0,\lambda]\bigr) \Bigr)=\frac{1}{\pi}\lim_{\delta\downarrow0}\int_{\lambda_0}^\lambda\operatorname{Im}\Sigma_g(\lambda+i\delta)\;\mathrm{d}\lambda.
	\label{eq:Stieltjes inversion}
		\end{equation}
		
		\item [(iii)] the boundary values of $\Sigma_g$ along $\operatorname{supp}(\kappa_g)$ are linked as follows to the Lebesgue decomposition of $\kappa_g$:
		\begin{enumerate}[$\bullet$]
			\item $\operatorname{supp}(\kappa_g)=\overline{\{\lambda\in\mathbb{R}\,|\,\operatorname{Im}\Sigma^+_g(\lambda)>0\}}$;

			\item $M_{\mathrm{ac}}=\{\lambda\in\operatorname{supp}(\kappa_g)\,|\,0<\operatorname{Im}\Sigma^+_g(\lambda)<\infty\}$ is a minimal support of $\kappa_g^{\mathrm{ac}}$, and the density of $\kappa_g^{\mathrm{ac}}$ is $\rho_g(\lambda)=\frac{1}{\pi}\operatorname{Im}\Sigma^+_g(\lambda)$;
			\item $M_{\mathrm{sc}} =\{\lambda\in\operatorname{supp}(\kappa_g)\,|\,\operatorname{Im}\Sigma^+_g(\lambda)=\infty,			\lim_{\delta\downarrow0}\delta\,\operatorname{Im}\Sigma_g(\lambda+i\delta)=0\}$
			is a minimal support of $\kappa_g^{\mathrm{sc}}$;
			\item  $\operatorname{supp}(\kappa_g^{\mathrm{pp}})=\overline{M}_{\mathrm{pp}}$, where
			$M_{\mathrm{pp}} = \{\lambda\in\operatorname{supp}(\kappa_g)\,|\,\operatorname{Im}\Sigma^+_g(\lambda)=\infty,\;\\
						\lim_{\delta\downarrow0}\delta\,\operatorname{Im}\Sigma_g(\lambda+i\delta)>0\}$, 
			and $\kappa_g(\{\lambda\})=\lim_{\delta\downarrow0}\delta\,\operatorname{Im}\Sigma_g(\lambda+i\delta)$,
			\end{enumerate}
			with $\Sigma^+_g$ defined in~\eqref{eq:Sigma+def}.
		
		\item [(iv)] Let $G_g(\lambda)$ be the function defined in~\eqref{eqn:gi}, then 
		\begin{enumerate}[$\bullet$]
			\item $\lim_{\delta\downarrow0}\frac{1}{\delta}\operatorname{Im}\Sigma_g(\lambda+i\delta)=G_g(\lambda)$;	
			\item in addition, if $G_g(\lambda)<\infty$, then $\Sigma^+_g(\lambda)$ is finite and real and $\Sigma_g(\lambda+i\delta)=\Sigma^+_g(\lambda)+i\delta G_g(\lambda)+o(\delta)$, as $\delta\downarrow0$,
			i.e.\  $G_g(\lambda)$ is the upper derivative of $\Sigma_g(z)$ in the direction of the imaginary axis.
			\end{enumerate}
		\end{itemize}
	\label{borel}
\end{proposition}
\begin{proof}
			By ~\eqref{eqn:nug}, the self-energy of the Friedrichs-Lee Hamiltonian $H_{g,\e}$ in Eq.~\eqref{self} can be written as
		\begin{equation}
		\Sigma_g(z)=\int_\R\left(\frac{1}{\lambda-z}-\frac{\lambda}{1+\lambda^2}\right)\mathrm{d}\kappa_g(\lambda).
		\end{equation}
	Therefore, $\Sigma_g(z)=B_{\kappa_g}(z)$, i.e.\  the self-energy is the regularised Borel transform of $\kappa_g$, and the properties $(i)$--$(iv)$ follow from the general theory of Borel transform: see e.g.~\cite{Gesztesy2} and references therein.
\end{proof}

\subsection{Cyclic subspaces and spectral properties}
In the previous subsection we have shown the link between the self-energy $\Sigma_g$ and the properties of $\kappa_g$; now we will link the latter with the spectral properties of the Friedrichs-Lee Hamiltonian. We will start from some basic definitions.
\begin{definition}
	Let $T$ be a self-adjoint operator on a Hilbert space $\mathcal{K}$, and let $\varphi\in\mathcal{K}_{-2}$. The cyclic subspace $\mathcal{K}_\varphi$ spanned by $\varphi$ is defined as follows:
	\begin{equation}
	\mathcal{K}_\varphi=\overline{\mathrm{Span}\left\{\frac{1}{T-z}\varphi\;\Bigg|\; z\in\mathbb{C}\setminus\mathbb{R}\right\}}.
	\label{eq:cyclics}
	\end{equation}
	In particular, if $\mathcal{K}_\varphi=\mathcal{K}$, $\varphi$ is called a cyclic vector and $T$ is said to have a simple spectrum.
\end{definition}
\begin{proposition}[\cite{DeOliveira,SimonRankOne2}]\label{propcyc}
	Let $T$ be a self-adjoint operator on a Hilbert space $\mathcal{K}$, and let $\varphi\in\mathcal{K}_{-2}$. The following properties hold:
	\begin{itemize}
		\item $\mathcal{K}_\varphi$ is a reducing subspace for $T$, and hence $T=T|_{\mathcal{K}_\varphi}\oplus T|_{\mathcal{K}_\varphi^\perp}$;
		\item $T|_{\mathcal{K}_\varphi}$ is unitarily equivalent to the position operator (i.e.\ the multiplication operator by the identity function) on the Hilbert space $L^2(\mathbb{R},\nu_\varphi)$, where $\nu_\varphi$ is the spectral measure of $T$ at $\varphi$, defined by $\nu_\varphi(B) = \braket{\varphi | E_T (B) \varphi}$, for all bounded Borel sets $B \subset \mathbb{R}$, with $E_T (\cdot)$ being the spectral projection of $T$.	
	 \item 
	$\sigma\left(T|_{\mathcal{K}_\varphi}\right)=\operatorname{supp}(\nu_\varphi)$ and  $\sigma_{j}\left(T|_{\mathcal{K}_\varphi}\right)=\operatorname{supp}(\nu^{j}_\varphi)$, for $j\in\{\mathrm{ac},\,\mathrm{sc},\,\mathrm{pp}\}$.
	\end{itemize}
\end{proposition}
\begin{remark}
	In our case, $\mathcal{K}=\mathcal{H}=L^2(X,\mu)$, $T=\Omega$ and $\varphi = g$. The cyclic subspace $\mathcal{H}_g$  spanned by $g$, as a consequence of Stone-Weierstrass theorem, is equal to  $L^2(X_g,\mu)$, with $X_g$ given by Eq.~\eqref{eqn:xg}; this justifies the use of the notation $\mathcal{H}_g=L^2(X_g,\mu)$ in Theorem~\ref{thm:spe}. In particular, $\Omega|_{\mathcal{H}_g}$ and $\Omega|_{\mathcal{H}_g^\perp}$ are the multiplication operators by the restriction of $\omega$ to $X_g$ and $X\setminus X_g$, respectively. Finally, $g$ is cyclic iff its support is a set of full measure $\mu$. 
\end{remark}
\begin{remark}
	As a consequence of Propositions~\ref{borel} and~\ref{propcyc}, the spectrum of $\Omega|_{\mathcal{H}_g}$ and its decomposition can be studied by analyzing the boundary values of the (regularised) Borel transform~\eqref{eqn:regbor} of the coupling measure $\kappa_g$.
\end{remark}
\begin{proposition}
	\label{propcyc2}
	Let $\e \in \R$, $g \in \mathcal{H}_{-2}$, $g \neq 0$, and let $H_{g,\e}$ be the corresponding Friedrichs-Lee Hamiltonian~\eqref{domain}--\eqref{action}. Then:
	\begin{enumerate}
		\item [(i)] the cyclic subspace of $\CC\oplus\mathcal{H}$ spanned by $\Psi_0$ in Eq.~\eqref{vacuum} is $\mathbb{C}\oplus\mathcal{H}_g$, with $\mathcal{H}_g$ being the cyclic subspace of $\mathcal{H}$ spanned by $g$. In particular, if $g$ is cyclic in $\mathcal{H}$, $\Psi_0$ is cyclic in $\CC\oplus\mathcal{H}$.
		\item [(ii)] $H_{g,\e}|_{(\CC\oplus\mathcal{H}_g)^\perp}=H_{0,\e}|_{(\CC\oplus\mathcal{H}_g)^\perp}$, and hence the spectrum of $H_{g,\e}|_{(\CC\oplus\mathcal{H}_g)^\perp}$ is the same as that of the uncoupled case.
	\end{enumerate}
\end{proposition}
\begin{proof}
Let $\mathcal{K}= \CC \oplus \mathcal{H}$. By definition, the cyclic subspace $\mathcal{K}_{\Psi_0}$ spanned by $\Psi_0$ is
\begin{equation}
\mathcal{K}_{\Psi_0}=\overline{\mathrm{Span}\left\{\frac{1}{H_{g,\e}-z}\Psi_0\;\bigg|\; z\in\mathbb{C}\setminus\mathbb{R}\right\}},
\end{equation}
where, by Eq.~\eqref{resolvent00}, for any $z\in\CC\setminus\R$ we have
\begin{equation}
	\frac{1}{H_{g,\e}-z}\Psi_0=\frac{1}{\e-z-\Sigma_g(z)}\begin{pmatrix}
1\\-\frac{1}{\Omega-z}g
	\end{pmatrix}.
	\label{eq:43}
\end{equation}
To prove that $\mathcal{K}_{\Psi_0}=\CC\oplus\mathcal{H}_g$, we will prove the equivalent equality $\mathcal{K}^\perp_{\Psi_0}=\left(\CC\oplus\mathcal{H}_g\right)^\perp=\{0\}\oplus\mathcal{H}_g^\perp$. 

Let $\Psi=\begin{pmatrix}
	x\\\xi
	\end{pmatrix}$ be orthogonal to vectors of the form~\eqref{eq:43}, i.e., such that
	\begin{equation}
	\frac{1}{\varepsilon-z-\Sigma_g(z)}\left[x-\Braket{\xi\,\bigg|\frac{1}{\Omega-z}g}\right]=0, \qquad\forall z\in\CC\setminus\R.
	\end{equation}
	This implies that $x=\Braket{\xi\,\Big|\frac{1}{\Omega-z}g}$ for all $z\in\CC\setminus\R$, and hence necessarily that
	\begin{equation}
	x=0,\qquad \Braket{\xi\,\bigg|\frac{1}{\Omega-z}g}=0\quad\forall z\in\CC\setminus\R,
	\end{equation}
	the second equality meaning that $\xi\in\mathcal{H}_g^\perp$. This proves $(i)$. Besides, by the expression of the Hamiltonian in Eq.~\eqref{action}, for any $\xi\in\mathcal{H}_g^\perp \cap D(\Omega)$ we have
\begin{equation}
	H_{g,\e}\begin{pmatrix}
	0\\\xi
	\end{pmatrix}=\begin{pmatrix}
	0\\\Omega\xi
	\end{pmatrix},
\end{equation}
so that $\{0\}\oplus\mathcal{H}_g^\perp$ is invariant under $H_{g,\e}$ (which is true in general for any cyclic subspace) and, in particular, the action of $H_{g,\e}$ on it is independent of $g$. This proves $(ii)$.
\end{proof}

\subsection{Proof of Theorem~\ref{thm:spe}}
\label{sec:spe}
We are now ready to prove Theorem~\ref{thm:spe}.
\begin{proof}
	By Proposition~\ref{propcyc2} we know that $H_{g,\e}|_{(\CC\oplus\mathcal{H}_g)^\perp}=H_{0,\e}|_{(\CC\oplus\mathcal{H}_g)^\perp}$, so we must only find the spectrum of the restriction of $H_{g,\e}$ to the cyclic subspace $\CC\oplus\mathcal{H}_g$. To simplify the notation, without loss of generality let us suppose $g$ cyclic for $\Omega$; then, Proposition~\ref{propcyc2} implies that $\Psi_0$ is cyclic for $H_{g,\e}$. Since both operators, $\Omega$ and $H_{g,\e}$, are self-adjoint on their Hilbert spaces, by Proposition~\ref{propcyc} the two operators are equivalent to the position operators on $L^2(\mathbb{R},\kappa_g)$ and $L^2(\mathbb{R},\nu_{\Psi_0})$, respectively, with $\nu_{\Psi_0}$ being the spectral measure of $H_{g,\e}$ at $\Psi_0$; finally, by Proposition~\ref{borel}, the $\mathrm{ac}, \mathrm{sc}, \mathrm{pp}$ components of the maximal spectral measures $\kappa_g$, $\nu_{\Psi_0}$ of both operators can be inferred by the boundary values of the imaginary parts of their Borel transforms.
	
	Now, the (regularised) Borel transform of $\kappa_g$ is the self-energy $\Sigma_g(z)$. The (standard) Borel transform of $\nu_{\Psi_0}$ is defined as follows:
	\begin{equation}
		\Pi_g(z)=\left\langle\Psi_0\,\bigg|\frac{1}{H_{g,\e}-z}\Psi_0\right\rangle,
	\end{equation}
	where we dropped the regularising term, which is not needed because $\Psi_0\in\CC\oplus\mathcal{H}$ and hence $\nu_{\Psi_0}$ is a finite measure. A straightforward calculation yields
	\begin{equation}
		\Pi_g(z)=\frac{1}{\e-z-\Sigma_g(z)},
	\end{equation}
	and hence, for any $\lambda\in\mathbb{R}$ and $\delta>0$,
	\begin{equation}
\operatorname{Im}\Pi_g(\lambda+i\delta)\sim\frac{\operatorname{Im}\Sigma_g(\lambda+i\delta)}{|\e-\lambda-\Sigma_g(\lambda+i\delta)|^2}, \quad \text{as } \delta\downarrow0.
\label{transforms}
	\end{equation}
	
	First of all, let us examine the absolutely continuous components. By Eq.~\eqref{transforms}, $\mathrm{Im}\,\Sigma_g(\lambda+i\delta)$ has a finite nonzero limit if and only if $\mathrm{Im}\,\Pi_g(\lambda+i\delta)$ has a finite nonzero limit; by Proposition~\ref{borel}, this proves that $\kappa^{\mathrm{ac}}_g$ and $\nu^{\mathrm{ac}}_{\Psi_0}$ are minimally supported on (Lebesgue-almost everywhere) equal sets; since they are absolutely continuous measures, this means that their densities are supported on (Lebesgue-a.e.) equal sets and hence the position operators on the corresponding $L^2$ spaces are unitarily equivalent. This proves the equality between the absolutely continuous spectra of $\Omega$ and $H_{g,\e}$.	
	
	Now let $\lambda$ be in a minimal support of the singular spectrum of $H_{g,\e}$; by Prop.~\ref{borel} and Eq.~\eqref{transforms} this happens iff $\e-\lambda=\Sigma^+_g(\lambda)$ (since, if $\operatorname{Im}\Sigma_g(\lambda+i\delta)$ diverges as $\delta\downarrow0$, the denominator diverges faster), also implying that $\Sigma^+_g(\lambda)$ is real. To distinguish between the pure point and singular continuous components, we must examine the limiting value of $\delta\,\Pi_g(\lambda+i\delta)$ as $\delta\downarrow0$. We have 
	\begin{equation}
	\delta\,\operatorname{Im}\Pi_g(\lambda+i\delta)\sim\frac{\frac{1}{\delta}\operatorname{Im}\Sigma_g(\lambda+i\delta)}{\frac{1}{\delta^2}\left(\e-\lambda-\operatorname{Re}\Sigma_g(\lambda+i\delta)\right)^2+\frac{1}{\delta^2}\operatorname{Im}\Sigma_g(\lambda+i\delta)^2}.
	\label{delta}
	\end{equation}
	Now, when $G_g(\lambda)=\infty$, the denominator in~\eqref{delta} diverges faster than the numerator and hence $\delta\,\Pi_g(\lambda+i\delta)\to0$ as $\delta\downarrow0$. Besides, when $G_g(\lambda)<\infty$, the first term in the denominator in~\eqref{delta} vanishes since $\operatorname{Re}\Sigma_g(\lambda+i\delta)\sim\e-\lambda+o(\delta)$ and hence $\left(\e-\lambda-\operatorname{Re}\Sigma_g(\lambda+i\delta)\right)^2\sim o(\delta^2)$; hence we are left with 
\begin{equation}
\delta\,\operatorname{Im}\Pi_g(\lambda+i\delta)\sim\frac{\delta}{\operatorname{Im}\Sigma_g(\lambda+i\delta)}\to\frac{1}{G_g(\lambda)}>0, \quad \text{as } \delta\downarrow0.
\end{equation}

As a result, the singular continuous component of $\nu_{\Psi_0}$ is minimally supported on the set
\begin{equation}
\left\{\lambda\in \R\,\big|\,\varepsilon-\lambda=\Sigma^+_g(\lambda),\,G_g(\lambda)=\infty\right\},
\end{equation}
and the support $\operatorname{supp}(\nu^{\mathrm{pp}}_{\Psi_0})$ of the pure point component of $\nu_{\Psi_0}$ is
\begin{equation}
\operatorname{supp}(\nu^{\mathrm{pp}}_{\Psi_0})=\overline{\left\{\lambda\in \R\,\big|\,\varepsilon-\lambda=\Sigma^+_g(\lambda),\,G_g(\lambda)<\infty\right\}}.
\end{equation}
\end{proof}
\begin{remark}
Notice that $\operatorname{Im}\Sigma_g^+(\lambda)=\infty$ iff $\lim_{\delta\downarrow0}\operatorname{Im}\Pi_g(\lambda+i\delta)=0$. In the case in which the singular spectra of both $\Omega$ and $H_{g,\e}$ are purely discrete, this means that inside the support of the spectral measure, the eigenvalues of $\Omega$ and $H_{g,\e}$ are completely disjoint: physically, no stable state of the field with energy coupled to the atom preserves its stability.

As for the energy $\e$ of the excited atom, which is an eigenvalue of the uncoupled operator:
\begin{itemize}
\item if $\e\notin\operatorname{supp}(\kappa_g)$, it will also be an eigenvalue of the coupled operator;
\item if $\e\in\operatorname{supp}(\kappa_g)$, then  it is an eigenvalue of the coupled operator if it satisfies the equation
\begin{equation}
\Sigma_g^+(\e)=0.
\end{equation}
In particular, if the spectral density of the absolutely continuous part of $\kappa_g$ is continuous, it must vanish at $\e$ .
\end{itemize}

\end{remark}

\begin{remark}
The eigensystem for the Friedrichs-Lee Hamiltonian can be solved explicitly. Let $\lambda\in\mathbb{R}$, $x\in\mathbb{C}$ and $\xi\in D(\Omega)$ such that
\begin{equation}\label{eq:eigen}
	H_{g,\e}\begin{pmatrix}
	x\\\xi-x\frac{\Omega}{\Omega^2+1}g
	\end{pmatrix}=\lambda\begin{pmatrix}
	x\\\xi-x\frac{\Omega}{\Omega^2+1}g
	\end{pmatrix};
\end{equation}
a direct calculation shows that $\lambda$ must solve Eq.~\eqref{eigeneq} and
\begin{equation}
\xi(k)=-x\left(\frac{1}{\omega(k)-\lambda}-\frac{\omega(k)}{\omega(k)^2+1}\right)g(k),
\end{equation}
where the condition on $\lambda$ ensures $\xi$ to be square-integrable.

\end{remark}

\begin{remark}
\label{rmk:couplmeas}
An important consequence of Theorem~\ref{thm:spe} is that, modulo a possible uncoupled part,  the spectrum of the Friedrichs-Lee Hamiltonian with a given $\e$ depends entirely on the coupling measure $\kappa_g$, or equivalently on the self-energy $\Sigma_g(z)$ which can be always reconstructed from the coupling measure through the inversion formula~\eqref{eq:Stieltjes inversion}.
Different choices of the momentum space $(X,\mu)$, of the dispersion relation $\omega$ and of the form factor $g$, but yielding the same $\kappa_g$ and the same uncoupled part, are fully equivalent at the spectral level. 

To this respect it is worth noticing that \emph{every} Borel measure on $\R$ satisfying the growth condition~\eqref{growth} can be obtained as the coupling measure of a Friedrichs-Lee Hamiltonian. Indeed, if $\nu$ is the desired coupling measure, just choose  $(X,\mu)=(\R,\nu)$,  $\omega(k)=k$ and $g(k)=1$, for all $k \in \R$. Then one immediately obtains
\begin{equation}
	\kappa_g(B)=\int_{\omega^{-1}(B)}|g(k)|^2\,\mathrm{d}\mu(k)=\int_B\mathrm{d}\nu(\lambda)=\nu(B),
\end{equation}
for any Borel set $B \subset \R$.
\end{remark}

\section{Some examples}  \label{sect:4}
In this section we discuss the spectral properties of some interesting examples of Friedrichs-Lee models. The section is organised as follows:
\begin{itemize}
\item Examples~\ref{ex1}--\ref{ex3} concern the case in which $\sigma(\Omega)$ is purely absolutely continuous;
\item Example~\ref{ex4} explores a purely discrete $\sigma(\Omega)$;
\item Finally, in Example~\ref{ex5} we investigate a pure point $\sigma(\Omega)$, with dense eigenvalues in $[0,1]$, which becomes singular continuous when the coupling is switched on.
\end{itemize}
Some considerations are now in order. Suppose that $\sigma(\Omega)$ is purely absolutely continuous in some (possibly unbounded) closed interval $J\subset\mathbb{R}$, and hence
\begin{itemize}
	\item $\sigma_{\mathrm{ac}}(H_{0,\varepsilon})=J$;
	\item $\sigma_{\mathrm{pp}}(H_{0,\varepsilon})=\{\varepsilon\}$,
\end{itemize}
where $\varepsilon$ can also be in $J$. Now, switching on a coupling $g$, we will still have $\sigma_{\mathrm{ac}}(H_{g,\e})=J$. In particular, $\varepsilon$ is again in the absolutely continuous spectrum of $H_{g,\varepsilon}$  if and only if $\varepsilon \in J$, but in general it will not  be in $\sigma_{\mathrm{pp}}(H_{g,\varepsilon})$.  However, $\varepsilon\in\sigma_{\mathrm{pp}}(H_{g,\varepsilon})$ if and only if $\Sigma^+_g(\varepsilon)=0$, i.e.\  if $\varepsilon$ is a zero for the  density (assumed to be continuous) of $\kappa_g$ and in addition $\operatorname{Re}\Sigma_g^+(\varepsilon)=0$. Physically, the eigenvalue becomes unstable whenever it lies in a set of coupled values of energy, except when the coupling density vanishes at that point. Of course, depending on the choice of $\kappa_g$, the pure point spectrum may contain other elements. In particular, if the coupling density is nonzero on the whole real line, the singular spectrum is empty and the eigenstate $\Psi_0$ in Eq.~\eqref{vacuum} of the uncoupled Hamiltonian $H_{0,\e}$ becomes unstable. 

More generally, if some $\lambda\in\mathbb{R}$ is a zero of the coupling density (and hence $\operatorname{Im}\Sigma^+_g(\lambda)=0$), a necessary and sufficient condition for it to be in the singular spectrum is that the equation $\varepsilon-\lambda=\operatorname{Re}\Sigma^+_g(\lambda)$ is fulfilled.

\begin{example}[Lebesgue coupling measure on $\mathbb{R}$]\label{ex1}
	Suppose that $J=\mathbb{R}$ and the coupling measure~\eqref{eqn:nug} reads
\begin{equation}
\mathrm{d}\kappa_g(\lambda)=\frac{\beta}{2\pi}\,\mathrm{d}\lambda
\end{equation}
for some $\beta>0$.  This is obtained e.g.\ from a Friedrichs-Lee model with dispersion relation $\omega(k)=k$  on $L^2(\R)$ and a flat form factor $g(k)=\sqrt{\beta/2\pi}$ (see Remark~\ref{rmk:couplmeas}).
Then a straightforward calculation shows that
\begin{equation}
	\Sigma_g(z)=\begin{cases}
	\frac{i\beta}{2},&\operatorname{Im}z>0;\\
	-\frac{i\beta}{2},&\operatorname{Im}z<0.
	\end{cases}
	\label{eq:flatSigma}
\end{equation}	
This implies that the pole equation~\eqref{eigeneq} does not have any solution and hence $\sigma(H_{g,\varepsilon})$ is purely absolutely continuous. The uncoupled eigenvalue  $\e$ ``dissolves'' in the continuum for  any nonzero value of the coupling $\beta$. Physically, the bound state with energy $\varepsilon$ becomes unstable. Indeed, one can show that the evolution of the state $\Psi_0$ in Eq.~\eqref{vacuum}, 
	\begin{equation}
	\begin{pmatrix}
	x(t)\\ \xi(t)
	\end{pmatrix}
	=U_{H_{g,\varepsilon}}(t) \Psi_0
	\end{equation}	
	yields $x(t)=e^{-\left(\frac{\beta}{2}+i\varepsilon\right)t}$, hence $|x(t)|^2=e^{-\beta t}$: an exponential decay takes place. Notice that a purely exponential decay law at both short and large times is possible since
	\begin{itemize}
		\item the initial state $\Psi_0$ is not in the domain of $H_{g,\varepsilon}$, notwithstanding it is in the domain of $H_{0,\varepsilon}$, since, being $\kappa_g$ Lebesgue, necessarily $g\in\mathcal{H}_{-2}\setminus\mathcal{H}$~\cite{artZeno}.
		\item $\Omega$ is unbounded both from below and from above, since $\sigma(\Omega)=\operatorname{supp}(\kappa_g)=\mathbb{R}$, and hence Paley-Wiener's theorem, which prohibits an exponential decay at large times, does not apply~\cite{PaleyWiener,Exner,Khalfin}.
	\end{itemize}
In Section~\ref{sect:5} we will interpret this result in the framework of the theory of resonances.
\end{example}
\begin{example}[Lebesgue coupling measure on $[0,\infty)$]\label{ex2}
	As a second example consider now a Friedrichs-Lee Hamiltonian with dispersion relation $\omega(k)=k$  on $\mathcal{H} = L^2(0,\infty)$ and a flat form factor $g(k)=\sqrt{\beta}$ for some $\beta>0$,  so that the coupling measure~\eqref{eqn:nug} reads
	\begin{equation}
\mathrm{d}\kappa_g(\lambda)=\beta\chi_{[0,\infty)}(\lambda)\,\mathrm{d}\lambda.
\end{equation} 

In this case the uncoupled spectrum is composed of the eigenvalue $\{\varepsilon\}$ and an absolutely continuous part in $J=[0,\infty)$.  Again the self-energy can be evaluated exactly:
\begin{equation}
\Sigma_g(z)=-\beta \log(-z),
\end{equation} 
with $\log$ being the principal value of the complex logarithm, i.e.\  $\log(-z)=\log|z|+i\operatorname{Arg}(-z)$, with $\operatorname{Arg}(z)\in(-\frac{\pi}{2}, \frac{\pi}{2})$. One can check that this function is indeed analytic in $\mathbb{C}\setminus[0,\infty)$ and has a branch cut along the support of the measure. Let us search for the singular spectrum when the coupling is switched on. Solutions of the pole equation~\eqref{eigeneq} must be searched in $(-\infty,0)$ since the coupling density is nonzero in $[0,\infty)$; we have
	\begin{equation}
\beta\log(-\lambda)=-\varepsilon+\lambda
	\end{equation} 
	and hence
	\begin{equation}
	-\frac{\lambda}{\beta}e^{-\lambda/\beta}=\frac{1}{\beta}e^{-\varepsilon/\beta},
	\end{equation} 
	which has a unique real solution expressed through the principal branch $W_0$ of Lambert's $W$-function (or product-log function)~\cite{specialfunct}:
	\begin{equation}
	E(\varepsilon,\beta)=-\beta\,W_0\left(\frac{1}{\beta}e^{-\varepsilon/\beta}\right).
	\end{equation} 
	Hence we have a unique eigenvalue $E(\varepsilon,\beta)$ for $H_{g,\varepsilon}$. It is interesting to study the asymptotic behavior of $E(\varepsilon,\beta)$ as a function of the excitation energy $\e$ of the atom. From the properties of the Lambert function,
	\begin{eqnarray}
	W_0(x)\sim \log x, \quad \text{as } x\to\infty, \qquad W_0(x)\sim x,\quad \text{as }  x\to0,
	\end{eqnarray}
	we have
	\begin{eqnarray}
	E(\varepsilon,\beta)\sim \varepsilon, \quad \text{as } \varepsilon\to-\infty;\qquad E(\varepsilon,\beta)\sim -e^{-\varepsilon/\beta},\quad\text{as } \varepsilon\to\infty.
	\end{eqnarray}
	This means that, when $\varepsilon$ is far away from the lowest energy level of the inner Hamiltonian $\Omega$, the coupled eigenvalue $E(\varepsilon,\beta)$ is close to $\varepsilon$  and hence the spectrum is nearly unchanged. When $\varepsilon$ approaches and eventually reaches $\sigma(\Omega)$, the approximation $E(\varepsilon,\beta)\sim\varepsilon$ is no longer valid and, as $\varepsilon\to\infty$,  $E(\varepsilon,\beta)$ approaches the boundary of $\sigma(\Omega)$.
\end{example}

\begin{example}[Sinusoidal coupling measure]\label{ex3}Now we consider a Friedrichs-Lee Hamiltonian in which the coupling density has support on the whole real line but admits some isolated zeros. Let the  coupling measure~\eqref{eqn:nug} be
\begin{equation}
\mathrm{d}\kappa_g(\lambda)=\frac{\beta}{2\pi}(1-\cos(\tau\lambda))\,\mathrm{d}\lambda,
\end{equation} 
for some $\beta>0$ and $\tau \in \R$,  i.e.\  an absolutely continuous measure with sinusoidal density. 
This can be obtained, e.g.\  from a Friedrichs-Lee model with dispersion relation $\omega(k)=k$ on $L^2(\R)$ and form factor $g(k)= (\beta (1-\cos \tau k )/ 2\pi)^{1/2}$.
As in Example~\ref{ex1}, here the uncoupled operator $H_{0,\e}$ has spectrum composed of the eigenvalue $\varepsilon$ embedded in an absolutely continuous spectrum covering the whole real line; however, here the coupling density vanishes at $\lambda_j=2j\pi/\tau$, with $j\in\mathbb{Z}$.
	
	The self-energy is
	\begin{equation}
	\Sigma_g(z)=\begin{cases}\frac{i\beta}{2}\left(1-e^{i\tau z}\right),&\operatorname{Im}z>0,\\\frac{-i\beta}{2}\left(1-e^{-i\tau z}\right),&\operatorname{Im}z<0,\end{cases}
	\end{equation}
	which is indeed discontinuous on the whole real line except for the zeros $\lambda_j$, moreover it is a periodic function. The pole equation~\eqref{eigeneq} reads
	\begin{equation}
	\begin{cases}
	\varepsilon-\lambda=\frac{\beta}{2}\sin\tau\lambda;\\
	\cos\tau\lambda=1.
	\end{cases}
	\end{equation}
	The second equation is satisfied only when, as anticipated, $\lambda$ is one of the zeros of the (continuous) coupling density; if so, the first equation simply becomes $\varepsilon-\lambda=0$. The following phenomenon occurs: the singular spectrum of $H_{g,\varepsilon}$ is empty except for some ``resonant" values of the parameter $\varepsilon$, namely $\varepsilon=\frac{2\pi j}{\tau}$ for some $j\in\mathbb{Z}$; when this happens, $\sigma_{\mathrm{pp}}(H_{g,\e})=\{\varepsilon\}$. 
\end{example}

\begin{example}[Periodic discrete coupling measure]\label{ex4}
	Consider a Friedrichs-Lee model with dispersion relation $\omega(k)=k$  and flat form factor  $g(k)= \sqrt{\beta/2\pi}$ ($\beta>0$), on $L^2(\R, \mu)$, where $\mu = \sum_{j\in\mathbb{Z}} \delta_{ j\tau}$ ($\tau>0$)
with $\delta_{k}$ being the Dirac measure at $k\in\mathbb{R}$. It is easy to show that the  coupling measure~\eqref{eqn:nug} is
\begin{equation}
\kappa_g =\frac{\beta}{2\pi}\sum_{j\in\mathbb{Z}} \delta_{ j\tau},
\end{equation} 
i.e.\  $\kappa_g$ is supported on $\tau\mathbb{Z}$. 
A direct calculation shows that
	$$
	\Sigma_g(z)=-\frac{\beta}{2}\cot\left(\frac{\pi z}{\tau}\right),
	$$
	which can be extended to the real line except for the poles at $\tau\mathbb{Z}$. The spectrum of the uncoupled operator will be purely singular and consisting of the solutions of the equation
	$$
	\cot\left(\frac{\pi \lambda}{\tau}\right)=\frac{2}{\beta}(\lambda-\varepsilon),
	$$
	which admits a countable set of isolated solutions $\{E_j(\e,\beta,\tau)\}_{j\in\mathbb{Z}}$ (hence the spectrum of $H_{g,\varepsilon}$ is again pure point) where $E_j(\e,\beta,\tau)\in(j\tau,(j+1)\tau)$. In particular, each $E_j(\e,\beta)$ varies smoothly with $\beta$ and
	\begin{itemize}
		\item $E_j(\e,\beta,\tau)\to \tau j$, as $\beta\to0$, i.e.\  in the limit of small coupling we recover the uncoupled spectrum of $\Omega$;
		\item $E_j(\e,\beta,\tau)\to\tau\left(j+\frac{1}{2}\right)$, as $\beta\to\infty$, i.e.\  in the limit of large coupling the spectrum is rigidly shifted by $\frac{\tau}{2}$.
	\end{itemize}
	Differently from the previous cases, the singular spectrum (which is again pure point) is nonempty for every value of the parameters, and indeed contains a countable number of points. 
\end{example}

\begin{example}[Generation of a singular continuous spectrum]\label{ex5}
This example is an adaptation of  Example 2 in~\cite{SimonWolf} to our framework. Recall that a real number in $[0,1]$ is said to be a dyadic rational of order $n\in\mathbb{N}$ if it can be written in the form $j/2^n$ for some integer $j$. Dyadic rationals of any order are the numbers whose expansion in base $2$ is finite; such numbers are  dense in $[0,1]$. 

For any integer $n \geq1$, let us consider the normalised Borel measure
\begin{equation}
\nu_n=\frac{1}{2^n}\sum_{j=1}^{2^n}\delta_{j/2^n},
\end{equation}
and a sequence of positive numbers $(a_n)_{n\in\mathbb{N}}$. We define the measure
\begin{equation}
\nu=\sum_{n=1}^\infty a_n\nu_n.
\end{equation}
$\nu$ is therefore a pure point measure with points on the dense set of dyadic rationals between $0$ and $1$; besides, it is a finite measure iff $\sum_na_n<\infty$, the latter sum being $\nu(\mathbb{R})=\nu([0,1])$.

Consider a Friedrichs-Lee model $H_{g,\e}$ with field momentum space $(X,\mu)=(\R,\nu)$, dispersion relation 
$\omega(k)=k$, and flat form factor $g(k)=\beta^{1/2}$, with $\beta>0$, so that the coupling measure reads $\kappa_g=\beta \nu$. Hence, the spectrum of the uncoupled Hamiltonian $H_{0,\e}$ is entirely pure point, consisting of the dyadic rationals in $[0,1]$ and (if not already dyadic) the atom excitation energy $\varepsilon$. By Theorem~\ref{thm:spe} we know that, by switching on the atom--field interaction, the new spectrum will be entirely singular (since creation of absolutely continuous spectrum is prohibited) and will be the closure of set of all solutions of the pole equation~\eqref{eigeneq}. The discriminant between singular continuous and pure point spectrum is given by the value of the function $G_g(\lambda)$ in~\eqref{eqn:gi}. In our case,
\begin{equation}
G_g(\lambda)=\int_\mathbb{R}\frac{1}{(\lambda-\lambda')^2}\,\mathrm{d}\kappa_g(\lambda')=\beta \sum_{n=1}^\infty a_nL_n(\lambda),
\end{equation}
with
\begin{equation}
L_n(\lambda)= \int_\mathbb{R}\frac{1}{(\lambda-\lambda')^2}\,\mathrm{d}\nu_n(\lambda')=\frac{1}{2^n}\sum_{j=1}^{2^n}\frac{1}{\left(\lambda-\frac{j}{2^n}\right)^2}.
\end{equation}
If $\lambda$ is dyadic, then $G_g(\lambda)=\infty$. Besides, for any $n$, there will be some integer $j_0\in\{1,\dots,n\}$ for which the quantity $|\lambda-j_02^{-n}|$ is the smallest among the others, i.e.\ $j_02^{-n}$ is the dyadic of order $n$ which is closest to $\lambda$. Therefore $|\lambda-j_02^{-n}|$ is smaller than the spacing between any two consecutive dyadic rationals of order $n$, which equals $2^{-n}$. This means that
\begin{equation}
\left|\lambda-\frac{j_0}{2^{n}}\right|\leq 2^{-n},
\end{equation}
and thus
\begin{equation}
\frac{1}{2^n}\frac{1}{\left(\lambda-\frac{j_0}{2^n}\right)^2}\geq\frac{1}{2^n}2^{2n}=2^n.
\end{equation}
In other words, for each $n$ the sum in the definition of $L_n(\lambda)$ contains one term which is larger than $2^n$; hence $L_n(\lambda)>2^n$ and thus
\begin{equation}
G_g(\lambda)\geq \beta \sum_{n=1}^\infty a_n2^n.
\end{equation}
This means that, if we choose $(a_n)_{n\in\mathbb{N}}$ in such a way that
\begin{equation}
\sum_{n=1}^\infty a_n2^n=\infty,
\end{equation}
then necessarily $G_g(\lambda)=\infty$ for \emph{every} $\lambda\in[0,1]$, meaning that the spectrum of $H_{g,\e}$ will be fully singular continuous. Interestingly,  the same phenomenon happens for every value of $\beta>0$, however small; this is an example of \emph{instability} of the dense set of eigenvalues under perturbations.
\end{example}

\section{Resonances}\label{sect:5}
We now complete the discussion on the spectral properties of the Friedrichs-Lee operator by studying the resonances of the model. In the following we will denote by $\CC^\pm := \{z\in\CC \, | \pm \operatorname{Im}z>0\}$ the open upper and lower half-planes, and for any subset $A \subset \CC$ we will use the notation $A^\pm :=A \cap \CC^\pm$ and  $A^0 :=A \cap \R$ for its components in $\CC^\pm$ and $\R$, respectively, so that $A$ is the disjoint union $A =  A^+ \cup A^- \cup A^0$.

Let us recall the definition of a resonance for a self-adjoint operator $T$.

\begin{definition}
	Let $T$ be a self-adjoint operator on a Hilbert space $\mathcal{K}$, and let $z_0\in\mathbb{C}$ with $\operatorname{Im}z_0<0$. Then $z_0$ is a \emph{resonance} for $T$ if there is some $\psi\in\mathcal{K}$ such that the function
	\begin{equation}\label{eqn:defres}
		z\in\CC^+\mapsto R_{\psi}(z)=\left\langle \psi\,\bigg|\frac{1}{T-z}\psi\right\rangle\in\CC
	\end{equation}
admits a meromorphic continuation from the upper to the lower half-plane having a pole at $z_0$.
\end{definition}
\begin{remark}
	The function~(\ref{eqn:defres}) is a Herglotz function; indeed, it is the (standard) Borel transform of the (finite) spectral measure of $T$ at the vector $\psi$, and hence the properties listed in Proposition~\ref{borel} hold true; in particular, every singularity of $R_{\psi}$ lies on the real line. In particular, if $\sigma(T)$ has an absolutely continuous component along some interval $J\subset\R$, then $R_{\psi}$ has a branch cut along $J$ with finite boundary values, and will thus admit an analytic continuation ``through the cut" from the upper to the lower plane~\cite{Gesztesy2}.
	
	A pole at $z_0$ of the meromorphic continuation is identified as a resonance, since it yields a contribution proportional to $e^{-iz_0t}$ to the survival amplitude of $\psi$; if $z_0$ is close to the real line, this contribution may dominate at large times and make $\psi$ a metastable state with energy $\operatorname{Re}z_0$ and decay rate $|\operatorname{Im}z_0|$.
\end{remark}
Now, for any $\psi\in\mathcal{K}$, \emph{real} poles of $R_{\psi}$, i.e.\ real resonances, are obviously simple poles of the resolvent of $T$, i.e.\ \emph{eigenvalues} of $T$. Conversely, an eigenvalue $\lambda_0$ of $T$ is \emph{not} necessarily a pole of $R_{\psi}$ for all nonzero $\psi\in\mathcal{K}$, since the spectral measure of  $T$  at $\psi$ may not be supported in a neighborhood of $\lambda_0$ (e.g., if $\psi$ is an eigenvector of $T$ with eigenvalue $\lambda\neq\lambda_0$); however, if we consider any dense subset $\mathcal{D}$ of $\mathcal{K}$, necessarily there will be some $\psi\in\mathcal{D}$ such that $R_{\psi}$ has a pole at $\lambda_0$. 
 
It would be useful to characterise \emph{non-real} resonances in a similar way. Aguilar-Balslev-Combes-Simon theory of resonances~\cite{AguilarCombes,BalslevCombes,SimonResonances} (see also~\cite{Hislop}) allows us to identify resonances of a self-adjoint operator $T$ as the (complex) eigenvalues of a ``deformed" Hamiltonian $T(w)$, with $w$ being a complex parameter.

We will apply this formalism to the Friedrichs-Lee Hamiltonian $H_{g,\e}$; for this purpose, we will need some assumptions on the structure of the inner Hamiltonian $\Omega$ and of the form factor $g$. First of all, we will make an assumption about the spectral properties of $\Omega$:
\begin{hypothesis}\label{hyp1}
	There is an interval $J\subset\R$ such that $J\subset \sigma_{\mathrm{ac}}(\Omega)$ and $J\cap \bigl(\sigma_{\mathrm{sc}}(\Omega) \cup \sigma_{\mathrm{pp}}(\Omega)\bigr)=\emptyset$.
\end{hypothesis}
\begin{remark}\label{rem:specess}
Notice that, by Theorem~\ref{thm:spessential} and Hypothesis~\ref{hyp1}, $J \subset \sigma_{\mathrm{ess}}(\Omega)$ and $\sigma_{\mathrm{dis}}(\Omega)\cap J= \emptyset$.
\end{remark}
Before introducing other hypotheses, we need the following definition:
\begin{definition}\label{defhyp}
	A \emph{spectral deformation family} is a family $(U(w))_{w\in W^0}$ of linear operators on $\mathcal{H}$ with the following properties
	\begin{itemize}
	        \item  $W^0 \subset \R$ is an open and connected neighbourhood  of $0$; 
		\item $U(w)$ is unitary  for all $w\in W^0$, and  $U(0)=I$;
		\item $(U(w))_{w\in W^0}$ admits a dense set $\mathcal{A}$ of \emph{analytic vectors}, i.e.\ such that for any $\psi\in\mathcal{A}$ the map $w\in W^0\mapsto \psi(w)= U(w)\psi$ has an analytic continuation in an open and connected set $W\subset\CC$, with $W^0=W\cap\R$; moreover $\mathcal{A}(w)=\{\psi(w) \,|\, \psi\in\mathcal{A}\}$ is dense in $\mathcal{H}$  for all $w \in W$.
	\end{itemize}
\end{definition}
Given a spectral deformation family $(U(w))_{w\in W^0}$, we  assume that for all $w\in W^0$:
\begin{equation}\label{eq:common}
U(w)D(\Omega)=D(\Omega),
\end{equation}
and define the deformation of the inner Hamiltonian $\Omega$ as
\begin{equation}\label{eq:deformationOmega}
	w\in W^0 \mapsto \Omega(w) \psi =U(w)\Omega U(w)^* \psi,
\end{equation}
for all $\psi\in D(\Omega)$.

More generally we will need that the operators $\Omega(w)$ are embedded in a \textit{type-A analytic} family, namely:
\begin{hypothesis}\label{hyp2}
There is a family of closed operators $(\Omega(w))_{w\in W}$, defined on the common domain $D(\Omega)$, such that, for all $\psi\in D(\Omega)$, the map 
\begin{equation}
w\in W\mapsto\Omega(w)\psi\in\mathcal{H},
\label{eq:deformationOmegatrue}
\end{equation} 
is the analytic continuation from $W^0$ into $W$ of the map~\eqref{eq:deformationOmega}.
\end{hypothesis}
Notice that for $w\in W^0$, $\Omega(w)$ is unitarily equivalent to $\Omega$ and, in particular, the two operators share the same spectrum; for non-real $w$, the spectrum of $\Omega(w)$ in general  differs from that of $\Omega$. For our purposes, we must require that the essential spectrum of $\Omega(w)$ in~\eqref{eq:deformationOmegatrue} changes ``nicely" as a function of $w\in W$, in the sense that it must be always possible, for any non-real $w$, to ``move continuously" through $J$ from the upper to the lower complex half-plane, so that the matrix elements of the resolvent of $\Omega(w)$ are analytic continuation of those of $\Omega$ through $J$. In order to accomplish these properties, we will make the following assumption:

\begin{hypothesis}
	\label{hyp3}
	There exists an open, connected subset $S\supset J$ of the complex plane
	 such that, for every
	$w\in W^+$, the following properties hold (see Figure~\ref{fig:deformation}):
	\begin{itemize}
	 \item $\sigma_{\mathrm{ess}}\left(\Omega(w)\right)\cap S^+=\emptyset$;
	 \item there exists some open connected $S_w\subset S$, with 
	 $S^0_w\neq\emptyset$,
	 such that
	 $\sigma_{\mathrm{ess}}\left(\Omega(w)\right)\cap (S^-_{w} \cup S^0_w)=\emptyset$.
	 \end{itemize}
\end{hypothesis}
In other words, we are requiring that the region $S$ is such that, for any $w \in W^+$,
\begin{itemize}
	\item the set $S^+$ does not contain elements of the essential spectrum of the deformed Hamiltonian $\Omega(w)$;
	\item there is a subset $S_w \subset S$ whose component with nonpositive imaginary part, i.e.\ $S^-_w \cup S_w^0$, does not intersect the essential spectrum of $\Omega(w)$. 
\end{itemize}

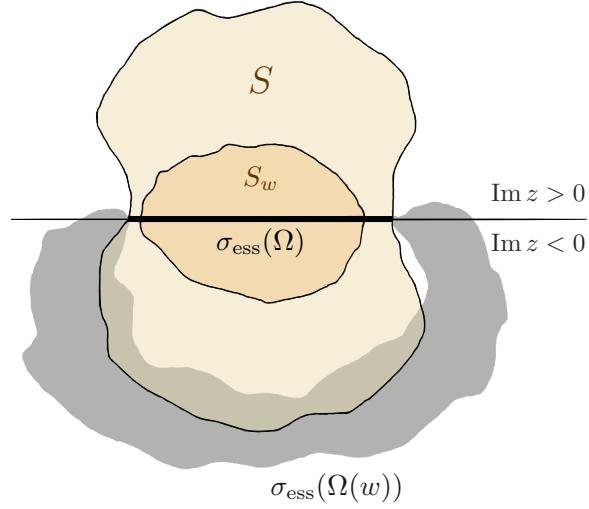
\begin{figure}\centering
	\begin{tikzpicture}
	\node at (0,0) {\includegraphics[scale=0.5]{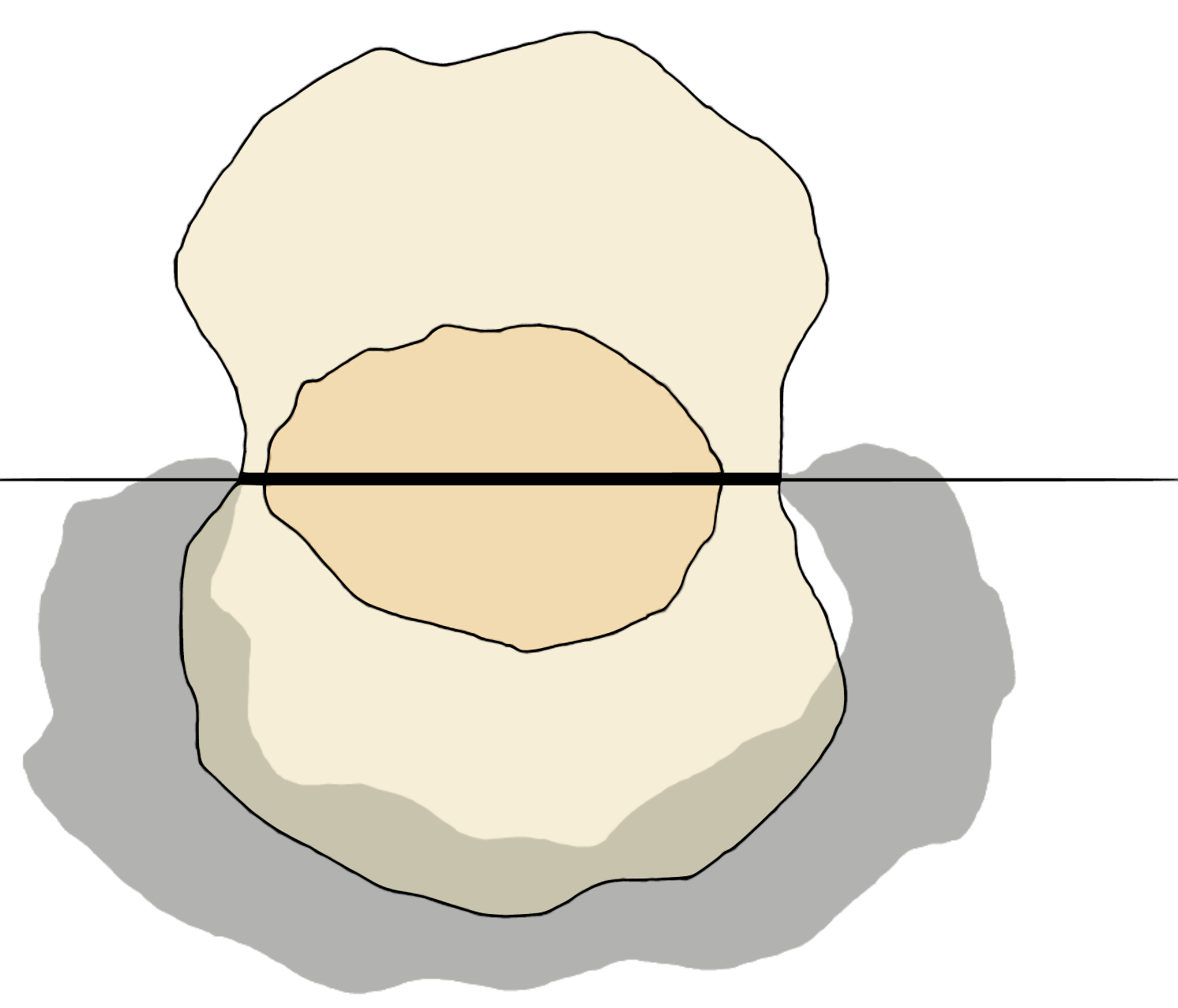}};
	\node[text=gray!30!black] at (3.15,0.50) {\small$\mathrm{Im}\,z>0$};
	\node[text=gray!30!black] at (3.15,-0.15) {\small$\mathrm{Im}\,z<0$};
	\node[text=orange!40!black] at (-0.5,2) {\Large$S$};
	\node[text=orange!40!black] at (-0.5,0.7) {$S_w$};
	\node[text=black] at (-0.5,-0.15) {$\sigma_{\mathrm{ess}}(\Omega)$};
	\node[text=gray!20!black] at (0.5,-3.4) {$\sigma_{\mathrm{ess}}(\Omega(w))$};
	\end{tikzpicture}
	\caption{Graphical representation of Hypothesis~\ref{hyp3}. For any $w \in W^+$ the set $S^+$ does not intersect the essential spectrum of the deformed Hamiltonian $\Omega(w)$ (grey region). Moreover,  there exists a set $S_w\subset S$ such $S^-_{w} \cup S^0_w$ does not intersect the essential spectrum of $\Omega(w)$.}
	\label{fig:deformation}
\end{figure}
Hypotheses~\ref{hyp1}--\ref{hyp3} would be sufficient  to characterise  all resonances of $\Omega$ as eigenvalues of its deformations $\Omega(w)$ (see~\cite{Hislop}); since we are interested in characterising the resonances of $H_{g,\e}$, we will need additional hypotheses on the form factor $g\in\mathcal{H}_{-2}$.

First of all, we define the following deformation of the form factor $g$:
\begin{equation}\label{eq:deformationg}
	w\in W^0 \mapsto \braket{g(w)|\psi}=\braket{g|U(w)^*\psi},
\end{equation}
for all $\psi\in D(\Omega)$, and require that the functionals $g(w)$ are embedded in an analytic family:
\begin{hypothesis}\label{hyp4}
	There exists a family $(g(w))_{w\in W}\subset\mathcal{H}_{-2}$ such that
for all $\psi\in D(\Omega)$, the map 
\begin{equation}
w\in W\mapsto {\braket{g(w)|\psi}}\in\CC
\end{equation} 
is the anti-analytic continuation from $W^0$ into $W$ of the map in~\eqref{eq:deformationg}.
\end{hypothesis}

Finally, we will need the following assumption on the form factor $g$:
\begin{hypothesis}\label{hyp5}
For all $\lambda \in J$ the boundary value $\Sigma_g^+$ in~\eqref{eq:Sigma+def} of the self-energy has positive imaginary part, namely
\begin{equation}
\operatorname{Im}\Sigma_g^+(\lambda)>0.
\end{equation}
\end{hypothesis}
\begin{remark}\label{specac}
Notice that, by Theorem~\ref{thm:spe}, Hypothesis~\ref{hyp1} and Hypothesis~\ref{hyp5}, it follows that $J\subset \sigma_{\mathrm{ac}}(H_{g,\e})$ and $J\cap \bigl(\sigma_{\mathrm{sc}}(H_{g,\e}) \cup \sigma_{\mathrm{pp}}(H_{g,\e})\bigr)=\emptyset$.
\end{remark}
Hypotheses~\ref{hyp4} and \ref{hyp5} on the form factor $g$, together with Hypotheses \ref{hyp1}--\ref{hyp3} on the inner Hamiltonian $\Omega$, allow us to characterise all resonances of the Friedrichs-Lee Hamiltonian $H_{g,\e}$.
	\begin{theorem}\label{thm:res}
		Let $H_{g,\e}$ a Friedrichs-Lee Hamiltonian~\eqref{domain}--\eqref{action} with excitation energy $\e\in\mathbb{R}$, inner Hamiltonian $\Omega$ and form factor $g\in\mathcal{H}_{-2}$. Suppose that Hypotheses~\ref{hyp1}--\ref{hyp5} hold true, and define, for all $w\in W$, the operator $H_{g,\e}(w)$ as follows:
			\begin{equation}
		\label{domainw}
		D(H_{g,\varepsilon}(w))=\left\{\begin{pmatrix}
		x\\\xi-x\frac{\Omega(w)}{\Omega(w)^2+1}g(w)
		\end{pmatrix}\Bigg|\,x\in\mathbb{C},\;\xi\in D(\Omega)\right\}
		\end{equation}
		\begin{equation}
		\label{actionw}
		H_{g,\varepsilon}(w)\begin{pmatrix}
		x\\\xi-x\frac{\Omega(w)}{\Omega(w)^2+1}g(w)
		\end{pmatrix}=\begin{pmatrix}
		\varepsilon x+\braket{g(w)|\xi}\\
		\Omega(w)\xi+x\frac{1}{\Omega(w)^2+1}g(w)
		\end{pmatrix}.
		\end{equation}
		Let $S^-_{W} = \bigcup_{w \in W^+}S^-_w$. Then:
		\begin{itemize}
			\item [(i)] There is a dense set $\tilde{\mathcal{A}} \subset \CC \oplus \mathcal{H}
			$ of analytic vectors such that, for every $\Psi\in\tilde{\mathcal{A}}$, the function
			 \begin{equation}
			 z\in\CC^+\mapsto \Braket{\Psi\,\bigg|\frac{1}{H_{g,\e}-z}\Psi}\in\mathbb{C}
			 \end{equation} 
			 has a meromorphic continuation across $J$ from $\CC^+$ to $S^-_{W}$;
			\item [(ii)]  $z_0 \in S^-_{W}$ is a resonance of $H_{g,\e}$ if and only if  $z_0$ is an eigenvalue of $H_{g,\e}(w_0)$ for some $w_0\in W^+$.\end{itemize}
\end{theorem}
	
	\begin{proof}
		Our proof will be organised as follows. First of all, we will prove that $H_{g,\e}$ satisfies properties analogous to the ones fulfilled by $\Omega$, more precisely:
		\begin{enumerate}
			\item [(a)] $H_{g,\e}$ has a purely absolutely continuous spectrum on $J$;
					\item [(b)] for all $w \in W^0$ we define
					\begin{equation}\label{eq:tilde}
					 \tilde{U}(w)=I\oplus U(w), \qquad \tilde{\mathcal{A}}=\mathbb{C}\oplus\mathcal{A};
					 \end{equation}
					 then it results that $(\tilde{U}(w))_{w\in W^0}$ is a spectral deformation family with a dense set of analytic vectors $\tilde{\mathcal{A}}$.
			Moreover, for all $w\in W^0$ the Hamiltonian $H_{g,\e}(w)$ in Eq.~\eqref{actionw} coincides with $\tilde{U}(w)H_{g,\e}\tilde{U}(w)^*$ and the map $w \in W^0 \mapsto H_{g,\e}(w)=\tilde{U}(w)H_{g,\e}\tilde{U}(w)^*$ admits an analytic continuation from $W^0$ into $W$. 
			\item [(c)] for every $w\in W^+$, $\sigma_{\mathrm{ess}}(H_{g,\e}(w))\cap S^+=\emptyset$ and $\sigma_{\mathrm{ess}}(H_{g,\e}(w))\cap (S_w^- \cup S_w^0 )=\emptyset$;
		\end{enumerate}
	then, following an analogous argument as in~\cite{Hislop}, we will deduce $(i)$ and $(ii)$ from $(a)$--$(c)$.
		
	Property $(a)$ follows directly from Theorems~\ref{thm:spessential} and Theorem~\ref{thm:spe} using Hypotheses~\ref{hyp1} and~\ref{hyp5}, as explained in Remarks~\ref{rem:specess} and~\ref{specac}. 
	
	Let us prove property $(b)$. It is straightforward that  $\tilde{\mathcal{A}}$ and $(\tilde{U}(w))_{w\in W}$ in~\eqref{eq:tilde} have the desired properties. To prove the equality $H_{g,\e}(w)=\tilde{U}(w)H_{g,\e}\tilde{U}(w)^*$ for all $w\in W^0$, notice that, for any $w\in W^0$, the domain of $\tilde{U}(w)H_{g,\e}\tilde{U}(w)^*$ is
	\begin{eqnarray}
			\label{domainw0}
			\tilde{U}(w)D(H_{g,\varepsilon})&=&\left\{\begin{pmatrix}
			x\\U(w)\xi'-xU(w)\frac{\Omega}{\Omega^2+1}g
			\end{pmatrix}\Bigg|\,x\in\mathbb{C},\;\xi'\in D(\Omega)\right\}\nonumber\\
			&=&\left\{\begin{pmatrix}
						x\\\xi-xU(w)\frac{\Omega}{\Omega^2+1}g
						\end{pmatrix}\Bigg|\,x\in\mathbb{C},\;\xi\in D(\Omega)\right\},
			\end{eqnarray}
			where we have used the assumption $U(w)D(\Omega)=D(\Omega)$ in~\eqref{eq:common}; besides,
			\begin{equation}
U(w)\frac{\Omega}{\Omega^2+1}g=\frac{\Omega(w)}{\Omega(w)^2+1}g(w),
			\end{equation}
			as follows easily from the equality $\Omega(w)=U(w)\Omega U(w)^*$ for  $w \in W^0$, and hence the domains \eqref{domainw} and \eqref{domainw0} coincide. The equality between the actions of $H_{g,\e}(w)$ and $\tilde{U}(w)H_{g,\e}\tilde{U}(w)^*$ on the domain can be shown in a similar way.

		Now we prove $(c)$ and $(d)$. By Hypothesis~\ref{hyp1}, for every $w\in W^+$, $\Omega(w)$ has no essential spectrum in $S^+$ and hence, by Hypotheses~\ref{hyp2}--\ref{hyp4}, then the function
		\begin{equation}
		(z,w)  \mapsto \Sigma_g(z,w)=\left\langle g(w)\,\bigg|\,\left(\frac{1}{\Omega(w)-z}-\frac{\Omega(w)}{\Omega(w)^2+1}\right)g(w)\right\rangle.
		\end{equation}
		is meromorphic in $S^+ \times (W^+ \cup W^0)$. Moreover, since $U(w)$ is unitary for all $w\in W^0$, it results that
		\begin{equation}
			\forall z\in S^+,\;\forall w\in W^0:\;\Sigma_g(z,w)=\Sigma_g(z).
		\end{equation}
	Hence, by the identity principle for meromorphic functions, we have
		\begin{equation}
		\forall z\in S^+,\;\forall w\in (W^+ \cup W^0):\;\Sigma_g(z,w)=\Sigma_g(z),
		\end{equation}
		and hence the function $z\in S^+ \times (W^+ \cup W^0) \mapsto \Sigma_g(z,w)$ is analytic because the self-energy $z  \in S^+\mapsto \Sigma_g(z)$ is analytic.    Now fix $w\in W^+$. By Hypothesis~\ref{hyp3} there is $S_w\subset S$ such that $\sigma_{\mathrm{ess}}(\Omega(w))$ has no intersection with the set $S^-_w \cup S^0_w$. Hence, for this value of $w$, the map $z \in S^+ \mapsto \Sigma_g(z,w)$ can be meromorphically continued from $S^+$ to $S^-_w$ across $S^0_w\subset J$. Repeating the process for every $w\in W^+$, the function $z \in S^+ \mapsto \Sigma_g(z,w)$ is meromorphically continued from $S^+$ to $S^-_{W} \cup S^0_{W}$, obtaining the function
		\begin{equation}\label{eq:meromcontsigma}
			z\in S^+\cup S^-_{W} \cup S^0_{W} \mapsto \Sigma_g(z,w)\in\CC;
		\end{equation}
this continuation is obviously unique and independent of $w$.
		
Now we compute the resolvent of $H_{g,\e}(w)$ with $w \in W^+$. It can be easily proven, in the same way as for the non-deformed Friedrichs-Lee Hamiltonian, that for every $z \in S^+\cup S^-_{W} \cup S^0_{W}$, with $z \notin \sigma(H_{g,\e}(w))$:
		\begin{equation}
		\frac{1}{H_{g,\e}(w)-z}\begin{pmatrix}
		x\\\xi
		\end{pmatrix}=\begin{pmatrix}
		\frac{x-\left\langle g(w)\,\big|\frac{1}{\Omega(w)-z}\xi\right\rangle}{\varepsilon-z-\Sigma_{g}(z,w)}\\\frac{1}{\Omega(w)-z}\xi-\frac{x-\left\langle g(w)\,\big|\frac{1}{\Omega(w)-z}\xi\right\rangle}{\varepsilon-z-\Sigma_{g}(z,w)}\frac{1}{\Omega(w)-z}g(w)
		\end{pmatrix}.
		\label{resolventw}
		\end{equation}
		Therefore the  elements of the spectrum of $H_{g,\e}(w)$ will be either
		\begin{itemize}
			\item elements of the spectrum of $\Omega(w)$, or
			\item solutions of the equation $\e-z=\Sigma_g(z,w)$.
		\end{itemize}
		By Hypothesis~\eqref{hyp3}, fixing $w\in W^+$, then $\sigma_{\mathrm{ess}}(\Omega(w))\cap S^+=\emptyset$ and $\sigma_{\mathrm{ess}}(\Omega(w))\cap (S^0_w\cup S^-_w)=\emptyset$, i.e.\ the region $S^+\cup S^0_w\cup S^-_w$ does not contain essential spectrum of $\Omega(w)$. Moreover, since the function in \eqref{eq:meromcontsigma} is meromorphic, 
the equation $\varepsilon-z=\Sigma_g(z,w)$ can only admit isolated solutions, i.e.\ the spectrum of $H_{g,\e}(w)$ can only have isolated eigenvalues in $S^-_w \cup S^0_w$.
	
		Now we can proceed with the proof of $(i)$ and $(ii)$.

		$(i)$ Let $\Psi\in\tilde{\mathcal{A}}$ and consider the following analytic function
		\begin{equation}
			z\in \CC^+ \mapsto R_\Psi(z)=\Braket{\Psi\,\bigg|\frac{1}{H_{g,\e}-z}\Psi}\in\CC.
		\end{equation}
		Since $(\tilde{U}(w))_{w \in W^0}$ is a spectral deformation family with dense set of analytic vectors $\tilde{\mathcal{A}}$, the map $w \in W^0 \mapsto \Psi(w)=\tilde{U}(w)\Psi \in \CC \oplus \mathcal{H}$ admits an  analytic continuation in $W$. 
		Now, we consider the function 
		\begin{equation}\label{eq:Rpsizw}
		(z,w)\in S^+\times (W^+ \cup W^0)  \mapsto R_\Psi(z,w)=\Braket{\Psi(w)\,\bigg|\frac{1}{H_{g,\e}(w)-z}\Psi(w)}\in\CC;
		\end{equation}
		using the expression for the resolvent in Eq.~\eqref{resolventw}, Hypotheses~\ref{hyp2}-\ref{hyp4}, $(b)$ and $(c)$, we obtain that the function in~\eqref{eq:Rpsizw}  is meromorphic. Moreover it is easy to see that, for all $w\in W^0$,
		\begin{equation}
\Braket{\Psi(w)\,\bigg|\frac{1}{H_{g,\e}(w)-z}\Psi(w)}=\Braket{\tilde{U}(w) \Psi\,\bigg| \tilde{U}(w) \frac{1}{H_{g,\e}-z}\Psi},
		\end{equation}
		and thus, since $\tilde{U}(w)$ is unitary,
		\begin{equation}
			\forall z\in S^+,\;\forall w\in W^0:\;R_{\Psi}(z,w)=R_\Psi(z).
		\end{equation}
		Hence, by the identity principle for meromorphic functions, we have
		\begin{equation}
		\forall z\in S^+,\;\forall w\in W^+ \cup W^0:\;R_{\Psi}(z,w)=R_\Psi(z),
		\end{equation}
		and so the function $z\in S^+ \times(W^+ \cup W^0)  \mapsto R_{\Psi}(z,w)$ is analytic because $z  \in S^+\mapsto R_{\Psi}(z)$ is analytic. Now fix $w\in W^+$. By $(c)$ we have $\sigma_{\mathrm{ess}}(H_{g,\e}(w))\cap (S_w^- \cup S_w^0 )=\emptyset$; hence, for this value of $w$, the map $z \in S^+ \mapsto R_{\Psi}(z,w)$ can be meromorphically continued from $S^+$ to $S^-_w$ across $S^0_w\subset J$. Repeating the process for every $w\in W^+$, the function $z \in S^+ \mapsto R_{\Psi}(z,w)$ can be meromorphically continued from $S^+$ to $S^-_{W} \cup S^0_{W}$ obtaining the function
		\begin{equation}\label{eq:meromcont}
			z\in S^+\cup S^-_{W} \cup S^0_{W} \mapsto R_{\Psi}(z,w)\in\CC;
		\end{equation}
this continuation is obviously unique and independent of $w$.
		
		$(ii)$ Let $z_0\in S^-_{W}$ be a resonance for $H_{g,\e}$, i.e.\ $z_0\in S^-_{w_0}$ for some $w_0\in W^+$.
 By definition, there is some $\Psi \in \CC \oplus \mathcal{H}$  such that the function
	\begin{equation}\label{eqn:defres1}
		z\in\CC^+\mapsto R_{\psi}(z)=\left\langle \Psi\,\bigg|\frac{1}{H_{g,\e}-z}\Psi\right\rangle\in\CC
	\end{equation}
	admits a meromorphic continuation from the upper to the lower half-plane having a pole at $z_0$. Since $\tilde{\mathcal{A}}$ is dense in $\CC \oplus \mathcal{H}$, we can assume $\Psi \in \tilde{\mathcal{A}}$. Then, by~\eqref{eq:meromcont},  the meromorphic continuation from the upper to the lower half-plane of~\eqref{eqn:defres1} is the expectation value of the resolvent of $H_{g,\e}(w_0)$ at $\Psi(w_0)$, i.e.\ $\left\langle \Psi(w_0)\,\Big|\frac{1}{H_{g,\e}(w_0)-z}\Psi(w_0)\right\rangle$, and the poles of the latter object are necessarily eigenvalues of $H_{g,\e}(w_0)$; hence $z_0$ is an eigenvalue of $H_{g,\e}(w_0)$.
		
		Conversely, suppose that   $z_0\in S^-_{W}$   is an eigenvalue of $H_{g,\e}(w_0)$ for some $w_0\in W^+$. Since $\tilde{\mathcal{A}}(w_0)= \left\{ \Psi(w_0)\,\big|\,\Psi \in \tilde{\mathcal{A}}\right\}$ is dense in $\CC \oplus \mathcal{H}$,  there must be necessarily some $\Psi\in\tilde{\mathcal{A}}$ such that the map in~\eqref{eq:meromcont} has a pole in $z_0$, and hence $z_0$ is a resonance of $H_{g,\e}$.
\end{proof}

\begin{example}\label{ex6}
	Let $\Omega$ be the multiplication operator by a continuous dispersion relation $\omega$ on a momentum space $(X,\mu)$ with an absolutely continuous measure $\mu$, so that Hypothesis \ref{hyp1} is trivially satisfied. A spectral deformation family may be constructed by considering an \textit{isometric global flow} $R_w$, for $w\in W^0$, i.e.\ an operator acting on the momentum space $X$ as follows:
\begin{equation}
(U(w)\psi)(k)=\mathcal{J}_w(k)^{1/2}\psi(R_w(k))
\end{equation}
for $\psi\in\mathcal{H}$, $\mathcal{J}_w(k)$ being the Jacobian of the transformation. We must then find some open and connected $W\subset\CC$, with $W^0=W\cap\R$, such that:
\begin{itemize}
	\item there exists a dense set $\mathcal{A}$ of analytic vectors such that $w\in W^0\mapsto \psi(w)=U(w)\psi\in\mathcal{H}$ has an analytic extension to the whole set $W$;
	\item the operator $\Omega(w)=U(w)\Omega U(w)^*$ defined, for $w\in W^0$, as
	\begin{equation}
	(\Omega(w)\psi)(k)=\omega(R_w(k))\psi(k)
	\end{equation}
	for all $\psi\in D(\Omega)$, admits a strongly analytic extension $(\Omega(w))_{w\in W}$ to $W$ (Hypothesis \ref{hyp2});
	\item the form factor $g(w)$ defined, for $w\in W^0$, as $g(w)=U(w)g$, admits an analytic extension to $W$ (Hypothesis \ref{hyp4}). 
\end{itemize}
This means that both the dispersion relation $\omega$, the form factor $g$ and every function $\psi\in\mathcal{A}$, which are functions of the \textit{real} variable $k$, must admit an extension as function of a \textit{complex} variable ranging in some open subset of the complex plane. Besides, the flow must be chosen in such a way that Hypothesis \ref{hyp3} holds, i.e.\ the spectrum of $\Omega(w)$ (i.e.\  its essential range) is deformed in the desired way to unearth the resonances.

As a concrete example, consider a Friedrichs-Lee operator on $L^2(\mathbb{R})$ with dispersion relation $\omega(k)=k$ (hence with $\sigma(\Omega)=\mathbb{R}$) and flat singular form factor $g(k)=\sqrt{\beta/2\pi}$, $\beta >0$; this is a realisation of a Friedrichs-Lee model with flat coupling measure studied in example~\ref{ex1}. Define, for any $w\in\R$,
\begin{equation}
R_w(k)=k-w,
\end{equation}
whose Jacobian for real $w$ is $1$; hence we obtain
\begin{equation}
(U(w)\psi)(k)=\psi(k-w),\qquad (\Omega(w)\psi)(k)=(k-w)\psi(k).
\end{equation}
The extension of $\Omega(w)$ to complex values of $w$, with $D(\Omega(w))=D(\Omega)$, is immediate; besides, 
\begin{equation}
	\sigma(\Omega(w))=\sigma_{\mathrm{ess}}(\Omega(w))=\R-i\operatorname{Im}w,
\end{equation}
and hence, for $w\in W^+$, the (essential) spectrum of $\Omega$ shifts rigidly into the lower half-plane; Hypothesis \ref{hyp2} holds with $S^+=\CC^+$ and $S^-_w=\{w'\in\CC\,|\, \operatorname{Im}w'<- \operatorname{Im}w<0\}$.

We must also find a set of analytic vectors $\mathcal{A}$ for $U(w)$. Let $\mathcal{A}$ be the set of functions $\psi\in L^2(\mathbb{R})$ with the form
\begin{equation}
	\psi(k)=P(k)e^{-ak^2},
	\label{analytic}
\end{equation}
where $P$ is a polynomial  and $a>0$; such functions are dense in $L^2(\mathbb{R})$, and for any $w\in\CC$ the function $\psi(w,\cdot)=\psi(\,\cdot\, -w)\in L^2(\R)$ is naturally defined, hence $\mathcal{A}$ is a dense set of vectors on which the action of $U(w)$ can be defined for all $w\in\CC$. Finally, in our simple case we have $g(w,k)=\sqrt{\beta/2\pi}$ for all $w\in\CC$, i.e.\ $g(w,\cdot)$ is again a constant function.

We have thus found a dense set of analytic vectors $\tilde{\mathcal{A}}=\CC\oplus\mathcal{A}$ such that, for every $\Psi\in\tilde{\mathcal{A}}$, the function $z\in\CC^+\mapsto R_\Psi(z)=\big\langle\Psi\big|\frac{1}{H_{g,\e}-z}\Psi\big\rangle\in\CC$ has an analytic continuation in $S^-_W=\CC^-$. By Theorem \ref{thm:res}, all poles of those functions are solutions of the equation
\begin{equation}
	\e-z=\Sigma_g(z,w)
\end{equation}
for some $w\in\CC^+$. In our case, taking $\operatorname{Im}w$ large enough, this equation reads
\begin{equation}
	\e-z=\frac{i\beta}{2}
\end{equation}
and hence we have a single resonance in $z=\e-\frac{i\beta}{2}$. The associated eigenvector $\Psi$ of $H_{g,\e}(w)$ is characterised by a boson wavefunction $\xi(w,\cdot)$ defined by
\begin{equation}
	\xi(w,k)=\sqrt{\frac{\beta}{2\pi}}\frac{1}{k-w-\e+\frac{i\beta}{2}}.
\end{equation}
By construction, the Fourier-Laplace transform of $R_\Psi(z)$, i.e.\ its survival amplitude, will be an exponential function with decay rate $\beta/2$. In particular, in the limit $\beta\to0$, its decay rate and boson component vanish and we recover the (real) eigenvalue $\lambda=\e$, with eigenvector $\Psi_0$, of the uncoupled Friedrichs-Lee Hamiltonian. 

Summing up, as anticipated in the discussion in Example \ref{ex2}, when switching on a flat coupling $g(k)=\sqrt{\beta/2\pi}$ in the Friedrichs-Lee Hamiltonian with $X=L^2(\R)$ and $\omega(k)=k$, the uncoupled eigenvalue $\lambda=\e$ is converted into a \textit{resonance} associated to an unstable state whose survival probability decays exponentially with rate $\beta$.
\end{example}

\section{Inverse problem: the case of exponential decay}\label{sect:6}
Theorem~\ref{thm:spe} allows us to fully characterise the spectrum, and hence the dynamics, of a Friedrichs-Lee Hamiltonian with some given momentum space $(X,\mu)$, dispersion relation $\omega$ and form factor $g$. On the other hand, we may want to solve the \emph{inverse problem}, i.e.\  finding a Friedrichs-Lee model with some \emph{given} spectrum, and hence some given dynamics.

As pointed out in Remark~\ref{rmk:couplmeas}, the spectrum of the Friedrichs-Lee Hamiltonian with a fixed $\e$ depends entirely on the coupling measure $\kappa_g$ in Eq.~\eqref{eqn:nug}, or equivalently on the self-energy $\Sigma_g$ in Eq.~\eqref{self}, and different choices of the momentum space $(X,\mu)$, of the dispersion relation $\omega$ and of the form factor $g$, yielding the same $\kappa_g$, are fully equivalent at the spectral level. 

A trivial solution (although not necessarily physically meaningful) always exists: just choose $\omega(k)=k$, $(X,\mu)=(\R,\kappa_g)$ and $g(k)=1$.
However, in the applications, the momentum space $(X,\mu)$ and the dispersion relation $\omega$ are fixed by the nature and the structure of the bosonic bath, while the form factor $g$ may be suitable engineered. Thus we are led to study the following inverse problem: \emph{for a given choice of $(X,\mu)$ and $\omega$, what choices of $g$ correspond to a Friedrichs-Lee Hamiltonian whose coupling measure $\kappa_g$ yields our desired dynamics?}

As an example, let us analyse the following problem: let us construct a Friedrichs-Lee operator such that the square modulus of $x(t)$, defined as in Eq.~\eqref{psiev}, decays with a purely exponential law. By Eq.~\eqref{f0t}, an exponential law is obtained if $\kappa_g$ is, up to a multiplicative constant, the Lebesgue measure on the whole real line. Indeed, if $\mathrm{d}\kappa_g(\lambda)=\frac{\beta}{2\pi}\mathrm{d}\lambda$, for some $\beta >0$, then by Eq.~\eqref{self} one readily obtains~\eqref{eq:flatSigma}, namely, 
\begin{equation}
	\Sigma_g(z)=\begin{cases}
	+\frac{i\beta}{2},&\operatorname{Im}z>0;\\
	-\frac{i\beta}{2},&\operatorname{Im}z<0,
	\end{cases}
\end{equation}
and hence, by Eq.~\eqref{f0t}, for $t>0$
\begin{equation}
	x(t)=e^{-(\beta/2+i\e)t},
\end{equation}
implying $|x(t)|^2=e^{-\beta t}$, for any value of $\e$. A Friedrichs-Lee model on a measure space $(X,\mu)$, with dispersion $\omega$ and form factor $g$, will thus yield an exponential decay if and only if
\begin{equation}
	\frac{\beta}{2\pi}(\lambda-\lambda_0)=\int_{\omega^{-1}\left([\lambda_0,\lambda]\right)}|g(k)|^2\,\mathrm{d}\mu(k),\quad \forall \lambda_0, \lambda \in \R,\; \lambda\geq\lambda_0.
	\label{inverse}
\end{equation}
Let us find some form factors $g$ satisfying Eq.~\eqref{inverse}, with space $X=\R^d$ and $\mu$ as the Lebesgue measure on $\R^d$, for different choices of the dispersion relation $\omega$.
\begin{example}
Consider a dispersion relation $\omega$ depending only on the projection of the momentum in some direction; without loss of generality, we fix a reference frame such that
\begin{equation}
	\omega(k_1,\dots,k_d)=w(k_1),
\end{equation}
with $w:\R\rightarrow\R$ being a differentiable and strictly increasing function. The latter hypothesis ensures
\begin{equation}
	\omega^{-1}\left([\lambda_0,\lambda]\right)=\left[w^{-1}(\lambda_0),w^{-1}(\lambda)\right]\times\R^{d-1},
\end{equation}
and hence Eq.~\eqref{inverse} becomes
\begin{equation}
\frac{\beta}{2\pi}(\lambda-\lambda_0)=\int_{w^{-1}(\lambda_0)}^{w^{-1}(\lambda)}f(k_1)\;\mathrm{d}k_1,
\label{inverse2}
\end{equation}
with
\begin{equation}
	f(k_1)=\int_{\R^{d-1}}|g(k_1,k_2,\dots,k_d)|^2\,\mathrm{d}k_2\dots\mathrm{d}k_d.
\end{equation}
Eq.~\eqref{inverse2} is satisfied iff $f(k_1)=\frac{\beta}{2\pi}w'(k_1)$, hence
\begin{equation}
	g(k_1,k_2,\dots,k_d)=e^{i\phi (k_1)}\sqrt{\frac{\beta}{2\pi}w'(k_1)}\,h(k_2,\dots,k_d),
\end{equation}
with $\phi :\R \to \R$ being an arbitrary phase term, and $h:\mathbb{R}^{d-1}\rightarrow\mathbb{C}$ a function satisfying
\begin{equation}
\int_{\R^{d-1}}|h(k_2,\dots,k_d)|^2 \, \mathrm{d}k_2\dots\mathrm{d}k_d=1.
\end{equation}
This result can be readily generalised to the case in which $w$ is piecewise monotonically increasing or decreasing: in this case the form factor is
\begin{equation}
g(k_1,k_2,\dots,k_d)=e^{i\phi (k_1)}\sqrt{\frac{\beta}{2\pi}|w'(k_1)|}\:h(k_2,\dots,k_d).
\end{equation}
\end{example}

\begin{example}
As a second example, consider a dispersion relation $\omega$ depending only on the modulus of the momentum. For any $k \in \R^d$ we can write $k=rn$, where $r =|k|$ and $n \in \mathbb{S}^1$, with $\mathbb{S}^1$ being the unit sphere. We assume that
\begin{equation}
	\omega(k)=w(r),
\end{equation}
with $w:\R^+\rightarrow\R$ being a differentiable and strictly increasing function. Again we obtain an equation analogous to Eq.~\eqref{inverse2}:
\begin{equation}
	\frac{\beta}{2\pi}(\lambda-\lambda_0)=\int_{w^{-1}(\lambda_0)}^{w^{-1}(\lambda)}f(r)\,r^{d-1}\;\mathrm{d}r,
	\label{inverse3}
\end{equation}
where
\begin{equation}
	f(r)=\int_{S^1}|g(k)|^2\;\mathrm{d}S(n).
\end{equation}
Eq.~\eqref{inverse3} is satisfied when $f(r)=\frac{\beta}{2\pi}w'(r)/r^{d-1}$, and hence the  form factor is
\begin{equation}
	g(k)=e^{i\phi (r)}\sqrt{\frac{\beta}{2\pi}\frac{w'(r)}{r^{d-1}}}h(n),
\end{equation}
again with $\phi:\R \to \R$ being an arbitrary phase term and $h:\mathbb{S}^1\rightarrow\CC$ satisfying
\begin{equation}
\int_{S^1}|h(n)|^2\;\mathrm{d}S(n)=1.
\end{equation}
This procedure can be generalised for a function $w$ piecewise monotonically increasing or decreasing, in this case we  obtain the form  factor,
\begin{equation}
g(k)=e^{i\phi (r)}\sqrt{\frac{\beta}{2\pi}\frac{|w'(r)|}{r^{d-1}}}h(n).
\end{equation}
\end{example}

\section*{Acknowledgements}
This work is partially supported by Istituto Nazionale di Fisica Nucleare (INFN) through the project ``QUANTUM'', and by the Italian National Group of Mathematical Physics (GNFM-INdAM).

\end{document}